\documentclass[a4paper,11pt,oneside,fleqn]{article}

\usepackage[mathscr]{eucal}
\usepackage{amsmath,amssymb,amsthm}
\usepackage{graphicx}

\allowdisplaybreaks
\makeatletter
\long\def\@makecaption#1#2{%
  \vskip\abovecaptionskip\footnotesize
  \sbox\@tempboxa{#1. #2}%
  \ifdim \wd\@tempboxa >\hsize
    #1. #2\par
  \else
    \global \@minipagefalse
    \hb@xt@\hsize{\hfil\box\@tempboxa\hfil}%
  \fi
  \vskip\belowcaptionskip}
\makeatother

\newcommand{\todo}[1][\null]{\ensuremath{\clubsuit}}

\newcommand{\noprint}[1]{}
\newcommand{\checked}[1][\null]{\ensuremath{\boldsymbol{\surd}}}

\newcommand{\p}{\partial}
\newcommand{\const}{\mathop{\rm const}\nolimits}

\newcommand{\sgn}{\mathop{\rm sgn}\nolimits}
\newcommand{\lsemioplus}{\mathbin{\mbox{$\lefteqn{\hspace{.77ex}\rule{.4pt}{1.2ex}}{\in}$}}}
\newcommand{\spanindex}{{\mbox{\tiny$\langle\,\rangle$}}}
\newcommand{\EqOrd}{r}

\newtheorem{theorem}{Theorem}
\newtheorem{lemma}[theorem]{Lemma}
\newtheorem{corollary}[theorem]{Corollary}
\newtheorem{proposition}[theorem]{Proposition}
\newtheorem*{problem*}{Problem}
{\theoremstyle{definition}
\newtheorem{definition}[theorem]{Definition}

\newtheorem{remark}[theorem]{Remark}
\newtheorem*{remark*}{Remark}
}

\begin{document}

\par\noindent {\LARGE\bf
Group classification of linear evolution equations
\par}

{\vspace{4mm}\par\noindent {\bf Alexander Bihlo$^\dag$ and Roman O.\ Popovych$^\ddag\, ^\S$,
} \par\vspace{2mm}\par}

{\vspace{2mm}\par\noindent {\it
$^{\dag}$~Department of Mathematics and Statistics, Memorial University of Newfoundland,\\ \phantom{$^{\dag}$~}St.\ John's (NL) A1C 5S7, Canada
}}
{\vspace{2mm}\par\noindent {\it
$^\ddag$~Wolfgang Pauli Institut, Oskar-Morgenstern-Platz 1, 1090 Wien, Austria
}}
{\vspace{2mm}\par\noindent {\it
$^{\S}$~Institute of Mathematics of NAS of Ukraine, 3 Tereshchenkivska Str., 01601 Kyiv, Ukraine
}}

{\vspace{2mm}\par\noindent {\it
\textup{E-mail:} abihlo@mun.ca, rop@imath.kiev.ua
}\par}

\vspace{4mm}\par\noindent\hspace*{10mm}\parbox{140mm}{\small
The group classification problem 
for the class of (1+1)-dimensional linear $\EqOrd$th order evolution equations 
is solved for arbitrary values of $\EqOrd>2$. 
It is shown that a related maximally gauged class of homogeneous linear evolution equations 
is uniformly semi-normalized with respect to linear superposition of solutions 
and hence the complete group classification can be obtained using the algebraic method. 
We also compute exact solutions for equations from the class under consideration 
using Lie reduction and its specific generalizations for linear equations.
}\par\vspace{2mm}

\section{Introduction}

The investigation of higher-order evolution equations has been the subject of a considerable body of literature in recent years. Such equations naturally occur in the study of real-world problems, including water waves and solitary waves~\cite{grim02a,ito80a,kich92a}, thin film models~\cite{bert00a,qu06a}, image processing~\cite{you00a}
as well as integrable models~\cite{olve77a}. 

While higher-order evolution equations typically arising in applications are nonlinear, there is substantial interest in studying higher-order linear evolution equations as well. In particular, the linearization of nonlinear evolution equations, which plays a key role in perturbation theory and stability analysis of these equations, leads to linear evolution equations of the same order with quite general variable coefficients.  
The study of a class of such linear equations within the framework of group analysis of differential equations is still a nontrivial problem 
since, in general, equation coefficients, which are interpreted as arbitrary elements of the class, are functions of several variables. 
Even the main problems of group analysis of differential equations--on Lie symmetries and on equivalence of equations--have been properly solved for (1+1)-dimensional evolution equations only in the case of order two~\cite{lie81Ay,moro03a,ovsi82Ay}. 

Although symmetry methods play a more important role in the study of nonlinear differential equations than of linear ones, 
there are many papers devoted to various aspects of symmetry analysis of general linear systems of differential equations and their specific classes. 
This includes, in particular, 
general constraints imposed by the linearity of a system of differential equations on its Lie and point symmetries~\cite{blum90a,gray14a,gray15a}; 
structure of algebras of generalized symmetries~\cite{shap92a};
a specific advanced method for generating new solutions from known ones using Lie symmetries \cite{broa99a,fush96a}; 
the description of conservation laws and potential symmetries of (1+1)-dimensional second-order linear evolution equations~\cite{popo08Ay} 
as well as reduction operators and nonclassical reductions of these equations \cite{fush92a,popo08a};
structure of Lie invariance algebras of linear systems of ordinary differential equations, group classification of such systems and admissible transformations between them \cite{boyk15a,camp11a,gonz88a,mele14a,mkhi15a,samo99a}.

In this paper we solve the group classification problem 
for the class of (1+1)-dimensional (in general, inhomogeneous) $\EqOrd$th order ($\EqOrd>2$) linear evolution equations of the form
\begin{equation}\label{eq:GeneralLinearEvolutionEquation}
 u_t=A^k(t,x)u_k + B(t,x),\quad A^\EqOrd\ne0. 
\end{equation}
Here and in the following, $\EqOrd$ is assumed to be an arbitrary but fixed integer greater than~2,  
and the summation over the repeated index~$k$ from $0$ to $\EqOrd$ is implied. 
$u_k=\partial_k u=\partial^k u/\partial x^k$, where, by definition, $u_0=u$. 
The functions $A^k=A^k(t,x)$ and $B=B(t,x)$ are smooth (e.g., analytic) functions of their arguments. 
The underlying field is real or complex. The consideration is within the local framework.

The group classification of equations of the form~\eqref{eq:GeneralLinearEvolutionEquation} for low values of~$\EqOrd$ was the subject of several investigations. 
The case $\EqOrd=1$ is trivial since then the class~\eqref{eq:GeneralLinearEvolutionEquation} is the orbit of the degenerate equation $u_t=0$ 
with respect to the corresponding equivalence group. 
The case $\EqOrd=2$ is specific among nontrivial values of~$\EqOrd$. It was exhaustively studied in the course of group classification of general second-order linear partial differential equations with two independent variables~\cite{lie81Ay,ovsi82Ay}. See also the review in~\cite[Section~2]{popo08Ay}. 
As specific and well studied, this case is excluded from the further consideration. 
Lie symmetries of third-order linear evolution equations were classified in~\cite{gung04a}. 
Most recently, fourth-order equations of the form~\eqref{eq:GeneralLinearEvolutionEquation} were considered in~\cite{huan12a}. 
The aim of this paper is to completely study Lie symmetries of linear evolution equations of arbitrary fixed order~$\EqOrd>2$ 
using more advanced methods of group classification.
This enhances the above results for orders three and four and extends them to an arbitrary (nontrivial) order.

Moreover, the class~\eqref{eq:GeneralLinearEvolutionEquation} is a subclass of the class of general (1+1)-dimensional evolution equations of order~$\EqOrd$.
In~\cite{maga93Ay}, Magadeev studied contact symmetries of such equations with $\EqOrd>1$ up to contact equivalence. It was proved that if an evolution equation is not linearizable by a contact transformation then its contact symmetry algebra is of dimension not greater than $\EqOrd+5$. Magadeev classified, up to contact equivalence, algebras of vector fields in the space of $t,x,u,u_x$, which can serve as contact symmetry algebras for some evolution equations. At the same time, he did not present the form of equations, which admit these algebras. It is obvious that equations that are linearizable by contact transformations admit infinite-dimensional contact symmetry algebras. It is known that contact transformations between fixed linear evolution equations are prolongations of point transformations~\cite{vane14a}. This is why the classification of contact symmetries of such equations up to contact equivalence degenerates to the classification of point symmetries of such equations up to point equivalence. As a result, the group classification of the class~\eqref{eq:GeneralLinearEvolutionEquation} completes Magadeev's studies of contact symmetries of evolution equations.  

The further consideration is the following. 
In Section~\ref{sec:AlgebraicGroupClassificationMethod} we present a brief review of the algebraic method of group classification. 
Special attention is paid to the version of this method for uniformly semi-normalized classes of differential equations, which is relevant for this paper.
Section~\ref{sec:EquivalenceGroupoidLinearEvolutionEqs} is devoted to 
the computation of the equivalence groupoids and the equivalence groups
of the class~\eqref{eq:GeneralLinearEvolutionEquation} and some of its gauged subclasses. 
We show that the class~\eqref{eq:GeneralLinearEvolutionEquation} 
and two of its subclasses singled out by gauging subleading coefficients, $A^\EqOrd=1$ and $(A^\EqOrd,A^{\EqOrd-1})=(1,0)$,
are normalized 
but the corresponding three subclasses of homogeneous equations 
are merely uniformly semi-normalized with respect to linear superposition of solutions. 
The group classification of the class~\eqref{eq:GeneralLinearEvolutionEquation} 
reduces to that of the subclass associated with the gauge $(A^\EqOrd,A^{\EqOrd-1},B)=(1,0,0)$ 
being maximal among general gauges that can be imposed using families of equivalence transformations,
and the property of the above specific uniform semi-normalization of this subclass justifies
using the above special version of the algebraic method for the group classification.
The preliminary analysis of Lie symmetries for equations from the class~\eqref{eq:GeneralLinearEvolutionEquation} 
is presented in Section~\ref{sec:DetEqsLieSymmetriesLinearEvolutionEqs}. 
In the main Section~\ref{sec:GroupClassificationLinearEvolutionEqs} 
we completely solve the group classification problem for the class~\eqref{eq:GeneralLinearEvolutionEquation}. 
Section~\ref{sec:ExactSolutionsLinEvolEqs} is devoted to the computation of a few exact solutions for equations from this class 
using Lie reduction and various symmetry-based methods that are specific for linear equations. 
In the final Section~\ref{sec:ConclusionLinearEvolutionEqs} we give some concluding remarks.

\section{Algebraic method of group classification}\label{sec:AlgebraicGroupClassificationMethod}

\looseness=-1
We briefly review some essential notions and results required for carrying out 
the group classification of the class~\eqref{eq:GeneralLinearEvolutionEquation} using the algebraic method. 
This method can be traced back to Lie's work on symmetries of ordinary differential equations. 
Its modern version, which is  in large parts based on results of~\cite{popo06Ay},
was explicitly formulated and applied for several models in~\cite{bihl11Dy,popo10Cy,popo10Ay}. 
See also~\cite{boyk15a,kuru16a} for further developments and applications of the algebraic method. 
Different techniques within the framework of the algebraic method of group classification 
were used in~\cite{basa01Ay,gagn93a,gaze92a,gung04a,huan12a,maga93Ay,mkhi15a,zhda99Ay}.

We denote by $\mathcal L_\theta$ a system of differential equations of the form $L(x,u^{(\EqOrd)},\theta(x,u^{(\EqOrd)}))=0$ 
in the~$n$ independent variables $x=(x_1,\dots,x_n)$ and the~$m$ dependent variables~$u=(u^1,\dots,u^m)$, 
were $L$ is a tuple of differential functions of~$u$. 
The short-hand notation~$u^{(\EqOrd)}$ is used for the tuple of derivatives of~$u$ with respect to $x$ up to order~$\EqOrd$ 
and by convention~$u$'s are included in~$u^{(\EqOrd)}$ as the zero-order derivatives. 
The tuple $\theta=(\theta^1(x,u^{(\EqOrd)}),\dots,\theta^k(x,u^{(\EqOrd)}))$ of parametric functions 
runs through the solution set~$\mathcal S$ of an auxiliary system 
$S(x,u^{(\EqOrd)},\theta^{(q)}(x,u^{(\EqOrd)}))=0$ 
of differential equations in~$\theta$, 
where $\theta^{(q)}$ denotes the partial derivatives of the arbitrary elements~$\theta$ up to order $q$ with respect to both~$x$ and~$u^{(\EqOrd)}$. 
It is usually also necessary to explicitly include some inequalities of the form 
$\Sigma(x,u^{(\EqOrd)},\theta^{(q)}(x,u^{(\EqOrd)}))\ne0$ (${}>0$,\ ${}<0$, \dots) in the auxiliary system. 
The \textit{class of (systems of) differential equations} $\mathcal L|_{\mathcal S}$ then consists of the parameterized systems~$\mathcal L_\theta$'s, 
for which~$\theta$ runs through the set~$\mathcal S$.

For the class of differential equations~\eqref{eq:GeneralLinearEvolutionEquation}, we have 
$n=2$, $m=1$ and the specific notation of independent variables $x_1=t$ and $x_2=x$.
The auxiliary system of equations for the tuple of arbitrary elements $\theta=(A^0,\dots,A^\EqOrd,B)$ consists of the equations
\[
 A^k_{u_\alpha}=0,\quad B_{u_\alpha}=0,\quad |\alpha|\leqslant r,
\]
where $\alpha=(\alpha_1,\alpha_2)$ is a multi-index, $\alpha_1,\alpha_2\in\mathbb N\cup\{0\}$, $|\alpha|=\alpha_1+\alpha_2$, and
$u_{\alpha}=\partial^{|\alpha|}u/\partial t^{\alpha_1}\partial x^{\alpha_2}$, 
which means that the arbitrary elements do not depend on derivatives of~$u$. 
Moreover, the auxiliary inequality is $A^\EqOrd\ne0$, which guarantees 
that equations from the class~\eqref{eq:GeneralLinearEvolutionEquation} are indeed of order~$\EqOrd$.

Of central importance in the group classification of differential equations is the study of point transformations relating two equations of the class to each other. The triple $(\theta,\tilde\theta,\varphi)$, where $\theta,\tilde\theta\in\mathcal S$ are arbitrary elements such that the associated systems $\mathcal L_\theta$ and $\mathcal L_{\tilde\theta}$ from the class~$\mathcal L|_{\mathcal S}$ are similar, and $\varphi$ is a point transformation of~$(x,u)$ mapping~$\mathcal L_\theta$ to~$\mathcal L_{\tilde\theta}$, is called an \emph{admissible transformation}. The set of admissible transformations of the class~$\mathcal L|_{\mathcal S}$, denoted by~$\mathcal G^\sim=\mathcal G^\sim(\mathcal L|_{\mathcal S})$, has a natural groupoid structure and is called the \emph{equivalence groupoid} of the class~$\mathcal L|_{\mathcal S}$. The computation of the equivalence groupoid and the study of its properties constitute a key step in the algebraic method of group classification.

Those point transformations in the space of independent and dependent variables and arbitrary elements 
that are projectable to the space of~$(x,u^{(\EqOrd')})$ for each $\EqOrd'=0,\dots,\EqOrd$, 
are compatible with the contact structure on the space of~$(x,u^{(\EqOrd)})$ 
and map every system from the class~$\mathcal L|_{\mathcal S}$ to a system from the same class 
are called (usual) \emph{equivalence transformations}. 
The equivalence transformations of~$\mathcal L|_{\mathcal S}$ constitute a Lie (pseudo)group, 
which is called the \emph{equivalence group}~$G^\sim$ of the class~$\mathcal L|_{\mathcal S}$. 
Each equivalence transformation~$\mathcal T\in G^\sim$ generates a family of admissible transformations of the class~$\mathcal L|_{\mathcal S}$, 
$G^\sim\ni\mathcal T\rightarrow\{(\theta,\mathcal T\theta,\pi_*\mathcal T)\mid \theta\in\mathcal S\}\subset\mathcal G^\sim$. 
Here $\pi$ denotes the projection of the space of $(x,u^{(\EqOrd)},\theta)$ to the space of equation variables only, 
$\pi(x,u^{(\EqOrd)},\theta)=(x,u)$, 
and thus the pushforward $\pi_*\mathcal T$ of~$\mathcal T$ by $\pi$ is just the restriction of~$\mathcal T$ to the space of~$(x,u)$.
In~this way, the equivalence group generates a subgroupoid of the equivalence groupoid. 

The infinitesimal generators of one-parameter groups of equivalence transformations constitute the \emph{equivalence algebra}~$\mathfrak g^\sim$ of the class~$\mathcal L|_{\mathcal S}$. The vector fields from~$\mathfrak g^\sim$ are defined on the space of~$(x,u^{(\EqOrd)},\theta)$ and are projectable to the spaces of~$(x,u^{(\EqOrd')})$ for each $\EqOrd'=0,\dots,\EqOrd$. The compatibility with the contact structure of the space of~$(x,u^{(\EqOrd)})$ means that the projection of any vector field of~$\mathfrak g^\sim$ to the space of~$(x,u^{(\EqOrd)})$ coincides with the $\EqOrd$th order prolongation of the corresponding projection to the space of~$(x,u)$.

The \emph{maximal point symmetry group}~$G_\theta$ of the system~$\mathcal L_\theta$ consists of 
the point transformations of $(x,u)$ that preserve the solution set of~$\mathcal L_\theta$.
Each transformation~$\varphi$ from~$G_\theta$ induces the admissible transformation $(\theta,\theta,\varphi)$ in the class~$\mathcal L|_{\mathcal S}$.
The common part~$G^\cap$ of all~$G_\theta$ is called 
the \emph{kernel} of maximal point symmetry groups of systems from~$\mathcal L|_{\mathcal S}$, 
$G^\cap:=\bigcap_{\theta\in\mathcal S} G_\theta$. 
The infinitesimal generators of one-parameter subgroups of the maximal point symmetry group~$G_\theta$ (resp.\ the kernel~$G^\cap$)
are vector fields in the space of~$(x,u)$, which span 
the \emph{maximal Lie invariance algebra}~$\mathfrak g_\theta$ of~$\mathcal L_\theta$ 
(resp.\ of the \emph{kernel invariance algebra}~$\mathfrak g^\cap$ of the class~$\mathcal L|_{\mathcal S}$).

The solution of the group classification problem for the class~$\mathcal L|_{\mathcal S}$ is to find all $G^\sim$-inequivalent values for~$\theta\in\mathcal S$ such that the associated systems, $\mathcal L_\theta$, admit maximal Lie invariance algebras, $\mathfrak g_\theta$, that are wider than the kernel invariance algebra~$\mathfrak g^\cap$. Further taking into account additional point equivalences between obtained cases, provided such additional equivalences exist, solves the group classification problem up to~$\mathcal G^\sim$-equivalence.

It is customary to work with the maximal Lie invariance algebra~$\mathfrak g_\theta$ rather than with the associated group~$G_\theta$, since the former can be readily computed using the Lie infinitesimal method. In particular, the infinitesimal invariance criterion states that a vector field 
$Q=\sum_{i=1}^n\xi^i(x,u)\p_{x_i}+\sum_{a=1}^m\eta^a(x,u)\p_{u^a}$ 
is an element of the maximal Lie invariance algebra~$\mathfrak g_\theta$ if and only if it is true that
$
 Q^{(\EqOrd)}L(x,u^{(\EqOrd)},\theta^{(q)}(x,u^{(\EqOrd)}))=0
$
on the manifold~$\mathcal L^\EqOrd_\theta$ defined by the system~$\mathcal L_\theta$ and its differential consequences in the $\EqOrd$th order jet space~$J^{(\EqOrd)}$. 
The $\EqOrd$th order prolongation~$Q^{(\EqOrd)}$ of the vector field~$Q$ is given by the general prolongation formula, 
see~\cite{olve86Ay} and Section~\ref{sec:DetEqsLieSymmetriesLinearEvolutionEqs}.

Splitting the equations implied by the infinitesimal invariance criterion with respect to the parametric derivatives of~$u$ leads to the determining equations for the components of~$Q$. For a class of differential equations~$\mathcal L|_{\mathcal S}$, there may be a subsystem of the determining equations that does not involve the tuple of arbitrary elements~$\theta$ and hence can be integrated directly. The remaining part of the determining equations explicitly contains arbitrary elements and is referred to as the \emph{classifying equations}. The purpose of group classification is the exhaustive investigation of the classifying equations. The direct integration (up to the equivalence generated by the corresponding equivalence group) of this part of the determining equations is usually only possible for classes of the simplest structure, e.g.\ classes involving only constants or functions of a single argument as arbitrary elements, see, e.g., examples in~\cite{ovsi82Ay}. Since most classes of interest in applications are of more complicated structure, different methods have to be used, which at least enhance the direct method~\cite{niki01Ay,vane09Ay,vane12Ay}. 

The most advanced classification techniques rest on the study of algebras of vector fields associated with systems from the class under consideration and constitute, in total, the \emph{algebraic method} of group classification. 
Although very effective in practice, the algebraic method requires a certain underlying structure of the class, which is conveniently expressed using various notions of normalization.
The class of differential equations~$\mathcal L|_{\mathcal S}$ is \emph{normalized} (in the usual sense) 
if the subgroupoid induced by the (usual) equivalence group~$G^\sim$ of~$\mathcal L|_{\mathcal S}$ 
coincides with the entire equivalence groupoid~$\mathcal G^\sim$ of~$\mathcal L|_{\mathcal S}$. 
The algebraic method of group classification is usually the method of choice to solve the complete group classification problem for a normalized class. 
If the equivalence groupoid~$\mathcal G^\sim$ is generated jointly by the equivalence group~$G^\sim$ 
and point symmetry groups of systems from~$\mathcal L|_{\mathcal S}$, 
i.e., for any $(\hat\theta,\check\theta,\varphi)\in\mathcal G^\sim$ there exist 
$\hat\varphi\in G_{\hat\theta}$, $\check\varphi\in G_{\check\theta}$ and $\mathcal T\in G^\sim$ such that 
$\check\theta=\mathcal T\hat\theta$ and $\varphi=\check\varphi(\pi_*\mathcal T)\hat\varphi$,
then the class~$\mathcal L|_{\mathcal S}$ is called \emph{semi-normalized}. 
Note that one of the symmetry transformations $\hat\varphi$ or $\check\varphi$ 
can always be assumed to be the identity transformation. 

To establish the normalization properties of the class~$\mathcal L|_{\mathcal S}$ it is necessary to compute its equivalence groupoid~$\mathcal G^\sim$. 
This is done using the direct method. Here one fixes two arbitrary systems from the class, 
$\mathcal L_\theta\colon L(x,u^{(\EqOrd)},\theta(x,u^{(\EqOrd)}))=0$ 
and $\mathcal L_{\tilde\theta}\colon L(\tilde x,\tilde u^{(\EqOrd)},\tilde\theta(\tilde x,\tilde u^{(\EqOrd)}))=0$, 
and aims to find the (nondegenerate) point transformations, $\varphi$: $\tilde x_i=X^i(x,u)$, $\tilde u^a=U^a(x,u)$, $i=1,\dots,n$, $a=1,\dots,m$, connecting them. 
For this, one changes the variables in the system~$\mathcal L_{\tilde\theta}$ by expressing the derivatives $\tilde u^{(\EqOrd)}$ in terms of $u^{(\EqOrd)}$ and derivatives of the functions $X^i$ and $U^a$ and substituting $X^i$ and $U^a$ for $\tilde x_i$ and $\tilde u^a$, respectively. 
The requirement that the resulting transformed system has to be satisfied identically for solutions of~$\mathcal L_\theta$ 
leads to the system of determining equations for the transformation components of~$\varphi$. 
Then, e.g., the class~$\mathcal L|_{\mathcal S}$ is normalized (in the usual sense) if 
the following conditions are satisfied: 
The transformational part~$\varphi$ of each admissible transformation 
does in fact not depend on the fixed initial value~$\theta$ of the arbitrary-element tuple 
and is hence appropriate for any initial value of the arbitrary-element tuple. 
Moreover, the prolongation of~$\varphi$ to the space of~$(x,u^{(\EqOrd)})$ 
and the further extension to the arbitrary elements according to the relation between~$\theta$ and~$\tilde\theta$ 
give a point transformation in the joint space of~$(x,u^{(\EqOrd)},\theta)$.

The recently introduced notion of uniformly semi-normalized classes~\cite{kuru16a} also plays 
an important role for group classification of various classes of differential equations, including the class~\eqref{eq:GeneralLinearEvolutionEquation}.
Let $\pi_*G^\sim$ denote the restriction of~$G^\sim$ to the space of~$(x,u)$,  
$\pi_*G^\sim=\{\pi_*\mathcal T\mid \mathcal T\in G^\sim\}$.

\begin{definition}\label{DefinitionOfUniformlySemi-normalizedClasses}
A class of differential equations~$\mathcal L|_{\mathcal S}$ 
with equivalence groupoid~$\mathcal G^\sim$ and (usual) equivalence group~$ G^\sim$
is called \emph{uniformly semi-normalized} 
with respect to the symmetry-sub\-group family $\mathcal N_{\mathcal S}=\{N_\theta\mid N_\theta\subseteq G_\theta,\theta\in\mathcal S\}$
if the following properties are satisfied:
\begin{enumerate}\itemsep=0ex
\item
Each~$N_\theta$ trivially intersects $\pi_*G^\sim$ only at the identity transformation.
\item
$N_{\mathcal T\theta}=(\pi_*\mathcal T)N_\theta(\pi_*\mathcal T)^{-1}$ 
for any $\theta\in\mathcal S$ and any $\mathcal T\in G^\sim$.
\item
For any $(\hat\theta,\check\theta,\varphi)\in\mathcal G^\sim$ there exist 
$\hat\varphi\in N_{\hat\theta}$, $\check\varphi\in N_{\check\theta}$ and $\mathcal T\in G^\sim$ such that 
$\check\theta=\mathcal T\hat\theta$ and $\varphi=\check\varphi(\pi_*\mathcal T)\hat\varphi$.
\end{enumerate}
\end{definition}

The implementation of the algebraic method to carry out group classification of the class~$\mathcal L|_{\mathcal S}$ 
involves the theorem on splitting symmetry groups in uniformly semi-normalized classes \cite[Theorem~2]{kuru16a}. 
It implies that for each $\theta\in\mathcal S$, $N_\theta$ is a normal subgroup of~$G_\theta$, 
$G^{\rm ess}_\theta:=G_\theta\cap\pi_*G^\sim$ is a subgroup of~$G_\theta$, 
and the group~$G_\theta$ is a semidirect product of~$G^{\rm ess}_\theta$ acting on~$N_\theta$, 
$G_\theta=G^{\rm ess}_\theta\ltimes N_\theta$.  
If the family $\mathcal N_{\mathcal S}$ is a priori known, 
then the study of point symmetries of systems from the class~$\mathcal L|_{\mathcal S}$ 
reduces to classifying the essential parts, $G^{\rm ess}_\theta$, of point symmetry groups, $G_\theta$. 
The infinitesimal counterparts of these results are true for the corresponding Lie algebras of vector fields~\cite{kuru16a}.  

The special kind of uniform semi-normalization, which is originally discussed in~\cite{kuru16a} and is also relevant for this paper, is related to linear superposition of solutions. We specify it for the case of single equations of a single dependent variable~$u=u^1$. The starting point for the consideration is a normalized class~$\mathcal L^{\rm inh}|_{\mathcal S^{\rm inh}}$ of linear (in general, inhomogeneous) differential equations of the form~$L(x,u^{(\EqOrd)},\theta^{(q)}(x))=\zeta(x)$, where the arbitrary elements~$\theta=(\theta^1,\dots,\theta^k)$ and~$\zeta$ depend only on~$x$, and the corresponding auxiliary system for the arbitrary elements does not constrain~$\zeta$. The class has to satisfy the following properties: (i) each system from~$\mathcal L^{\rm inh}|_{\mathcal S^{\rm inh}}$ is locally solvable; (ii) the only common solution of homogeneous systems from~$\mathcal L^{\rm inh}|_{\mathcal S^{\rm inh}}$ is the zero solution; (iii) restricting the elements of the equivalence group~$G^\sim_{\rm inh}=G^\sim(\mathcal L^{\rm inh}|_{\mathcal S^{\rm inh}})$ to the space of equation variables~$(x,u)$ yields fiber-preserving transformations whose components for~$u$ are affine in~$u$, i.e.\ $\tilde x_i=X^i(x)$, $u=M(x)(u+h(x))$, where~$\det(X^i_{x_{i'}})\ne0$ and $M\ne0$. Here~$i$ and~$i'$ take values from~$1$ to~$n$, and~$h(x)$ is an arbitrary smooth function of~$x$. 

Each equation from the class~$\mathcal L^{\rm inh}|_{\mathcal S^{\rm inh}}$ can be mapped to the associated homogeneous equation using equivalence transformations. Moreover, the equivalence group~$G^\sim_{\rm inh}$ is split as~$G^\sim_{\rm inh}=H^\sim_{\rm inh}\ltimes N^\sim_{\rm inh}$, where~$H^\sim_{\rm inh}$ is the subgroup constituted by elements of~$G^\sim_{\rm inh}$ with $h=0$ and~$N^\sim_{\rm inh}$ is the normal subgroup of the transformations with $(x,u)$-components of the form~$\tilde x_j=x_j$ and~$\tilde u=u+h(x)$. The restriction of~$H^\sim_{\rm inh} $ to the space of~$(x,u^{(\EqOrd)},\theta)$ coincides with the equivalence group~$G^\sim_{\rm hmg}$ of the associated class of homogeneous equations, denoted by~$\mathcal L^{\rm hmg}|_{\mathcal S^{\rm hmg}}$. Since equations from~\smash{$\mathcal L^{\rm inh}|_{\mathcal S^{\rm inh}}$} are $G^\sim_{\rm inh}$-equivalent if and only if their homogeneous counterparts are $G^\sim_{\rm hmg}$-equivalent, the group classification of equations from~$\mathcal L^{\rm inh}|_{\mathcal S^{\rm inh}}$ can be obtained by classifying symmetries of equations from~\smash{$\mathcal L^{\rm hmg}|_{\mathcal S^{\rm hmg}}$}. This is the strategy that is employed in the present paper for solving the classification problem for the class~\eqref{eq:GeneralLinearEvolutionEquation}; see also Remark~\ref{rem:GroupClassificationOfLinearHomogeneousGaugeEvolEqs} below. 

The group classification of the class~$\mathcal L^{\rm hmg}|_{\mathcal S^{\rm hmg}}$ is facilitated by taking into account that this class is \emph{uniformly semi-normalized with respect to linear superposition of solutions}, which means the following: For each $\theta\in\mathcal S^{\rm hmg}$, consider the subgroup~$G^{\rm lin}_\theta$ of the point symmetry group~$G_\theta$ of the system~$\mathcal L_\theta$ from the class~$\mathcal L^{\rm hmg}|_{\mathcal S^{\rm hmg}}$ consisting of the linear superposition transformations, $\tilde x_j=x_j$ and $\tilde u=u+h(x)$, where~$h$ is a solution of~$\mathcal L_\theta$. The class~$\mathcal L^{\rm hmg}|_{\mathcal S^{\rm hmg}}$ is uniformly semi-normalized with respect to the family of subgroups~$\mathcal N_{\rm lin}=\{G^{\rm lin}_\theta\mid\theta\in\mathcal S^{\rm hmg}\}$. 

In view of the above theorem on splitting symmetry groups in uniformly semi-normalized classes, 
for each $\theta\in\mathcal S^{\rm hmg}$
the group~$G_\theta$ is decomposed as~$G_\theta=G^{\rm ess}_\theta\ltimes G^{\rm lin}_\theta$, 
where $G^{\rm ess}_\theta=G_\theta\cap\pi_*G^\sim_{\rm hmg}$ 
is the \textit{essential part of the symmetry group}~$G_\theta$. 
The splitting of~$G_\theta$ induces the splitting of the corresponding maximal Lie invariance algebra 
$\mathfrak g_\theta=\mathfrak g^{\rm ess}_\theta \lsemioplus \mathfrak g^{\rm lin}_\theta$. 
Here $\mathfrak g^{\rm ess}_\theta$ is the \emph{essential Lie invariance algebra} of~$\mathcal L_\theta$, 
$\mathfrak g^{\rm ess}_\theta=\mathfrak g_\theta\cap\pi_*\mathfrak g^\sim_{\rm hmg}$, 
and the ideal~$\mathfrak g^{\rm lin}_\theta$ consists of the vector fields 
that generates one-parameter symmetry groups related to linear superposition of solutions. 
Such ideals are a priori known and can be neglected in the course of group classification. 
Therefore, the group classification for the class of homogeneous differential equations~$\mathcal L^{\rm hmg}|_{\mathcal S^{\rm hmg}}$ can be carried out by classifying appropriate subalgebras of the equivalence algebra~$\mathfrak g^\sim_{\rm hmg}$ (resp., of its pushforward $\pi_*\mathfrak g^\sim_{\rm hmg}$ by the projection~$\pi$ to the space of $(x,u)$). 
Here we should adapt the definition of appropriate subalgebras, which was first introduced in~\cite{card11Ay}, 
to the case of classes that are uniformly semi-normalized with respect to linear superposition of solutions.
A subalgebra~$\mathfrak s$ of the algebra~$\mathfrak g^\sim_{\rm hmg}$ 
(resp., of its pushforward $\pi_*\mathfrak g^\sim_{\rm hmg}$)
is called \emph{appropriate} if its pushforward $\pi_*\mathfrak s$ by~$\pi$ 
(resp., $\mathfrak s$ itself)
coincides with~$\mathfrak g^{\rm ess}_\theta$ for some~$\theta\in\mathcal S^{\rm hmg}$.

\section{Equivalence groupoid}\label{sec:EquivalenceGroupoidLinearEvolutionEqs}

We compute the equivalence groupoids and equivalence groups of the class~\eqref{eq:GeneralLinearEvolutionEquation} and its subclasses using the direct method. 
The computation technique is similar to that for classes of single $\EqOrd$th order linear ordinary differential equations~\cite{boyk15a}.
Following~\cite{vane14a}, we first construct a nested chain of normalized superclasses for the class~\eqref{eq:GeneralLinearEvolutionEquation} 
starting from the widest convenient class of general (1+1)-dimensional $\EqOrd$th order evolution equations,
\begin{equation}\label{eq:GeneralEvolutionEquation}
 u_t=H(t,x,u_0,\dots,u_\EqOrd), \quad H_{u_\EqOrd}\ne0,
\end{equation}
and sequentially narrow this class to the class~\eqref{eq:GeneralLinearEvolutionEquation} 
by setting more constraints for the arbitrary element~$H$. 
This leads to the sequential restriction of the equivalence group of the class~\eqref{eq:GeneralEvolutionEquation}, 
which is manifested in setting more constraints for functions parameterizing equivalence transformations.
After reaching the class~\eqref{eq:GeneralLinearEvolutionEquation} 
and reparameterizing it using $\theta=(A^0,\dots,A^\EqOrd,B)$ as arbitrary elements instead of $H$, 
we continue with gauging of arbitrary elements of the class~\eqref{eq:GeneralLinearEvolutionEquation} 
by families of equivalence transformations 
to obtain a class appropriate for the group classification by the algebraic method. 
For each of the above steps, we present both the corresponding additional constraints for arbitrary elements 
and the constrained form of transformations. 

A contact transformation relates two fixed equations from the class~\eqref{eq:GeneralEvolutionEquation} if and only if they are of the form 
$\tilde t=T(t)$, $\tilde x=X(t,x,u,u_x)$, $\tilde u=U(t,x,u,u_x)$, $\tilde u_{\tilde x}=V(t,x,u,u_x)$ and 
$\tilde u_{\tilde t}=(U_t+U_uu_t-(X_t+X_uu_t)V)/T_t$,
where the usual nondegeneracy assumption and contact condition are satisfied~\cite{maga93Ay}. 
Subsequently adding the constraint~$H_{u_\EqOrd u_l}=0$, where $l=2,\dots,\EqOrd$, one then obtains that all contact transformations between equations from the corresponding subclass are induced by point transformations. In other words, the transformational part of admissible transformations is of the form $\tilde t=T(t)$, $\tilde x=X(t,x,u)$ and $\tilde u=U(t,x,u)$. If, furthermore, $H_{u_\EqOrd u_1}=0$, then $X=X(t,x)$ only. 
The next step is to impose $H_{u_\EqOrd u_0}=0$, which gives the constraint $U_{uu}=0$ for admissible transformations. The associated class is again normalized and still contains the class~\eqref{eq:GeneralLinearEvolutionEquation} of linear evolution equations. 
Hence we can start the computation of the equivalence groupoid of the class~\eqref{eq:GeneralLinearEvolutionEquation} with the obtained restricted form of point transformations
\[
\varphi\colon\quad \tilde t=T(t),\quad \tilde x=X(t,x),\quad \tilde u=U^1(t,x)u+U^0(t,x).
\] 
The nondegeneracy condition reduces to $T_tX_xU^1\ne0$. The equivalence groupoid is found upon fixing two arbitrary equations $\mathcal L_\theta$ and $\mathcal L_{\tilde\theta}$ from the class~\eqref{eq:GeneralLinearEvolutionEquation} and supposing their connection through a nondegenerate point transformation~$\varphi$ of the restricted form. In practice, this is done by re-expressing the jet variables~$(\tilde t,\tilde x,\tilde u^{(\EqOrd)})$ in terms of $(t,x,u^{(\EqOrd)})$, followed by substitution in~$\mathcal L_{\tilde\theta}$, which gives the intermediate equation~$\tilde{\mathcal L}$. Note that for the transformation~$\varphi$, the transformed derivative operators are
\[
 \partial_{\tilde t}=\frac{1}{T_t}\left(\partial_t-\frac{X_t}{X_x}\partial_x\right),\quad \partial_{\tilde x}=\frac{1}{X_x}\partial_x.
\]
Since~$\mathcal L_\theta$ and~$\mathcal L_{\tilde\theta}$ are assumed to be connected by a nondegenerate point transformation, 
the intermediate equation~$\tilde{\mathcal L}$ is then required to be identically satisfied by all solutions of~$\mathcal L_\theta$ for each appropriate transformation. 
Substituting the expression for~$u_t$ from~\eqref{eq:GeneralLinearEvolutionEquation} into~$\tilde{\mathcal L}$ and
splitting the resulting equation with respect to parametric derivative~$u_0$, \dots, $u_\EqOrd$, 
we derive only formulas connecting $\theta$ and~$\tilde\theta$ with no constraints for~$T$, $X$, $U^1$ and~$U^0$. 
Using Fa\`{a} di Bruno's formula, the explicit transformation formulae for the arbitrary elements could be obtained. However, they are quite cumbersome and are in fact not needed at this stage. On the other hand, the transformations of $A^\EqOrd$ and~$B$ are readily derived without the need to invoke Fa\`{a} di Bruno's formula,
\[
 \tilde A^\EqOrd=\frac{(X_x)^\EqOrd}{T_t}A^\EqOrd,\quad \tilde B=\frac{U^1}{T_t}\left(B+(\partial_t-A^k\partial_k)\frac{U^0}{U^1}\right),
\]
see also the analogous formulae for the class of (1+1)-dimensional linear second-order evolution equations presented in~\cite{popo08Ay}. 

The same transformation~$\varphi$ can be applied to any equation~$\mathcal L_\theta$ from the class~\eqref{eq:GeneralLinearEvolutionEquation} 
and maps~$\mathcal L_\theta$ to an equation~$\mathcal L_{\tilde\theta}$ from the same class. 
The relation between the arbitrary-element tuples~$\theta$ and~$\tilde\theta$,  
which is derived from the equation~$\tilde{\mathcal L}$, 
defines, when $\theta$ varies, 
the prolongation of the transformation~$\varphi$ to the arbitrary elements~$A^k$ and~$B$. 
The prolonged transformation is a point transformation in the joint space of~$(t,x,u,\theta)$ 
and thus belongs to the equivalence group~$G^\sim_{\mbox{\tiny\eqref{eq:GeneralLinearEvolutionEquation}}}$
of the class~\eqref{eq:GeneralLinearEvolutionEquation}.%
\footnote{%
Both the arbitrary elements $\theta=(A^0,\dots,A^\EqOrd,B)$ of the class~\eqref{eq:GeneralLinearEvolutionEquation} 
and the corresponding components of equivalence transformations 
do not depend on derivatives of~$u$. 
Therefore, the restriction of equivalence transformations to the space of~$(t,x,u,\theta)$ is well defined. 
By definition, the components of equivalence transformations for derivatives of~$u$ are expressed via 
the $t$-, $x$- and $u$-components. 
This is why it is convenient to assume 
that the equivalence group~$G^\sim_{\mbox{\tiny\eqref{eq:GeneralLinearEvolutionEquation}}}$ 
of the class~\eqref{eq:GeneralLinearEvolutionEquation} 
as well as the equivalence groups of its subclasses act in the space of~$(t,x,u,\theta)$.
} 
Since this exhausts all possible transformations among equations from the class~\eqref{eq:GeneralLinearEvolutionEquation}, 
the equivalence groupoid~$\mathcal G^\sim_{\mbox{\tiny\eqref{eq:GeneralLinearEvolutionEquation}}}$ of this class 
is induced by its equivalence group~$G^\sim_{\mbox{\tiny\eqref{eq:GeneralLinearEvolutionEquation}}}$. 
In other words, the class~\eqref{eq:GeneralLinearEvolutionEquation} is normalized.

Using equivalence transformations, we can gauge some of the arbitrary elements of the class~\eqref{eq:GeneralLinearEvolutionEquation}. For example, it would be possible to apply the gauge $B=0$ and thus obtain the general class of $\EqOrd$th order homogeneous linear evolution equations. The problem with this gauge is that the resulting class is not normalized anymore but only uniformly semi-normalized with respect to linear superposition of solutions. This is why the gauge $B=0$, which is essential for the efficient solution of the group classification problem, will be applied at the latest possible stage. Similarly, we could set $A^0=0$ or $A^0=A^1=0$, cf.~\cite{boyk15a,popo08Ay}, but the corresponding subclasses are of complicated structure with respect to point transformations. It is thus advantageous to apply the gauge $A^\EqOrd=1$, which singles out a normalized subclass of linear evolution equations. 
Assuming this gauge, we obtain $(X_x)^\EqOrd=T_t$, which implies $X_{xx}=0$ and therefore $X=X^1(t)x+X^0(t)$, where $(X^1)^\EqOrd=T_t$. 
The nondegeneracy condition then reduces to $T_tU^1\ne0$. 
The condition $X_{xx}=0$ essentially simplifies the computation of the transformed derivatives $\tilde u_k$.
Since $X_x$ does not depend on~$x$, we can now keep off the use of Fa\`{a} di Bruno's formula. 
Thus, the transformation of $A^{\EqOrd-1}$ gives, thanks to the general Leibniz rule, 
\[
 \tilde A^{\EqOrd-1}=\frac{U^1}{(X^1)^{\EqOrd-1}}A^{\EqOrd-1}+\EqOrd\frac{U^1_x}{(X^1)^\EqOrd}.
\]
The subsequent gauge $A^{\EqOrd-1}=0$ gives a normalized subclass,
\begin{equation}\label{eq:LinearInhomogeneousGaugeEvolEqs}
 u_t=u_\EqOrd+A^l(t,x)u_l+B(t,x).
\end{equation}
Here and in the following the summation over the repeated index~$l$ from $0$ to $\EqOrd-2$ is implied. 
Admissible transformations within the class~\eqref{eq:LinearInhomogeneousGaugeEvolEqs} satisfy the additional constraint $U^1_x=0$. 
This last simplification makes it possible to obtain 
the compact transformation formulae for all arbitrary elements~$A^l$ and~$B$, 
which proves the following theorem:

\begin{theorem}\label{thm:EquivalenceGroupInhomogenousEvolEqs}
\begin{subequations}\label{eq:EquivalenceGroupoidGenLinEvolutionEqs}
The class~\eqref{eq:LinearInhomogeneousGaugeEvolEqs} of reduced (1+1)-dimensional linear inhomogeneous evolution equations of order~$\EqOrd$ is normalized. 
Its equivalence group consists of the transformations of the form
\begin{align}
 &\tilde t=T(t),\quad \tilde x=X^1(t)x+X^0(t),\quad \tilde u=U^1(t)u+U^0(t,x),\label{eq:PointTransformationBetweenEvolEqs}\\
 &\tilde A^j=\frac{(X^1)^j}{T_t}A^j,\quad \tilde A^1=\frac{X^1}{T_t}A^1-\frac{X^1_tx+X^0_t}{T_t},\quad \tilde A^0=\frac{1}{T_t}\left(A^0+\frac{U^1_t}{U^1}\right),\label{eq:PointTransformationBetweenEvolEqsAE1}\\
 &\tilde B=\frac{U^1}{T_t}\left(B+(\partial_t-\partial_\EqOrd-A^l\partial_l)\frac{U^0}{U^1}\right),\label{eq:PointTransformationBetweenEvolEqsAE2}
\end{align}
\end{subequations}
where $j=2,\dots,\EqOrd-2$, $T=T(t)$, $X^0=X^0(t)$, $U^1=U^1(t)$ and $U^0=U^0(t,x)$ are arbitrary smooth functions of their arguments with $T_tU^1\ne0$, 
$X^1=\sqrt[\EqOrd]{T_t}$ if $\EqOrd$ is odd and $X^1=\epsilon\sqrt[\EqOrd]{T_t}$ with $\epsilon=\pm 1$ and $T_t>0$ if $\EqOrd$ is even.
\end{theorem}

Due to~$U^0$ being an arbitrary function in~\eqref{eq:EquivalenceGroupoidGenLinEvolutionEqs}, 
we can gauge the inhomogeneity $B$ to zero, 
leading to the important subclass of (1+1)-dimensional homogeneous linear $\EqOrd$th order evolution equations of the form
\begin{equation}\label{eq:LinearHomogeneousGaugeEvolEqs}
 u_t=u_\EqOrd+A^l(t,x)u_l.
\end{equation}
In the following, by $A$ and $\mathcal L_A$ we denote the tuple of arbitrary elements $(A^0,\dots,A^{\EqOrd-2})$ 
and the corresponding equation from the class~\eqref{eq:LinearHomogeneousGaugeEvolEqs}. 
In contrast to the associated class~\eqref{eq:LinearInhomogeneousGaugeEvolEqs} of inhomogeneous equations,
the class~\eqref{eq:LinearHomogeneousGaugeEvolEqs} loses the normalization property.  
Therefore, not the entire equivalence groupoid~$\mathcal G^\sim$ of the class~\eqref{eq:LinearHomogeneousGaugeEvolEqs} 
is induced by the equivalence group~$G^\sim$ of this class and it is necessary to describe both objects. 

\begin{corollary}
 The equivalence groupoid~$\mathcal G^\sim$ of the class~\eqref{eq:LinearHomogeneousGaugeEvolEqs} 
 of reduced (1+1)-dimensional linear homogeneous evolution equations of order~$\EqOrd$ consists of the triplets of the form~$(A,\tilde A,\varphi)$, 
 where the point transformation~$\varphi$ in the space of variables is of the form~\eqref{eq:PointTransformationBetweenEvolEqs} 
 with $U^0/U^1$ satisfying the equation~$\mathcal L_A$, 
 and the tuples of arbitrary elements~$A$ and~$\tilde A$ are related by~\eqref{eq:PointTransformationBetweenEvolEqsAE1}.
\end{corollary}

\begin{proof}
If $B=0$ and $\tilde B=0$, then the equation~\eqref{eq:PointTransformationBetweenEvolEqsAE2} implies that the ratio $U^0/U^1$ is a solution of~$\mathcal L_A$.
\end{proof}

\begin{corollary}
 The usual equivalence group~$G^\sim$ of the class~\eqref{eq:LinearHomogeneousGaugeEvolEqs} of (1+1)-dimensional linear homogeneous evolution equations of order~$\EqOrd$ consists of transformations of the form~\eqref{eq:PointTransformationBetweenEvolEqs} and~\eqref{eq:PointTransformationBetweenEvolEqsAE1} with $U^0=0$.
\end{corollary}

\begin{proof}
Since equivalence transformations are point transformations in the space of variables $(t,x,u,A)$ that can be applied to all equations from the class, only those transformations of the form~\eqref{eq:PointTransformationBetweenEvolEqs} and~\eqref{eq:PointTransformationBetweenEvolEqsAE1} for which $U^0$ runs through the set of common solutions of this class satisfy this requirement. The unique common solution of equations from class~\eqref{eq:LinearHomogeneousGaugeEvolEqs} is the zero solution $u=0$, which is readily seen, e.g., from subtracting the two equations~$u_t=u_\EqOrd$ and $u_t=u_\EqOrd+u$. Hence $U^0=0$ for equivalence transformations.
\end{proof}

\begin{corollary}
The class~\eqref{eq:LinearHomogeneousGaugeEvolEqs} is uniformly semi-normalized with respect to linear superposition of solutions.  
\end{corollary}

\begin{proof}
The assertion is implied by the following facts: 
Each equation from the class~\eqref{eq:LinearHomogeneousGaugeEvolEqs} is locally solvable. 
The unique common solution of equations from class~\eqref{eq:LinearHomogeneousGaugeEvolEqs} is the zero solution $u=0$, cf. the previous proof. 
For a general admissible transformation in the class~\eqref{eq:LinearHomogeneousGaugeEvolEqs}, 
its transformational part is of the form~\eqref{eq:PointTransformationBetweenEvolEqs}.
Hence the transformation component(s) for $(t,x)$ (resp.\ $u$) do not depend on~$u$ (resp.\ is affine in~$u$), 
the transformation parameters~$T$, $X^0$, $X^1$ and $U^1$ do not depend on the arbitrary elements $A^l$'s, 
and the ratio $U^0/U^1$ runs through the solution set of the corresponding initial equation. 
\end{proof}

\begin{corollary}\label{cor:OnEquivalenceAlgebraOfHomogeneousEquations}
The equivalence algebra of the class~\eqref{eq:LinearHomogeneousGaugeEvolEqs} of (1+1)-dimensional linear homogeneous evolution equations of order~$\EqOrd$ is given by
\begin{equation}\label{eq:EquivalenceAlgebraGenLinEvolutionEqs}
 \mathfrak g^\sim=\langle\hat D(\tau),\hat P(\chi),\hat I(\phi)\rangle,
\end{equation}
where $\tau$, $\chi$ and $\phi$ run through the set of smooth functions of~$t$, with
\begin{align*}
 &\hat D(\tau)=\tau\partial_t+\frac{1}{\EqOrd}\tau_tx\partial_x-\sum_{j=2}^{\EqOrd-2} \frac{\EqOrd-j}{\EqOrd}\tau_tA^j\p_{A^j}-\left(\frac{\EqOrd-1}{\EqOrd}\tau_t+\frac{1}{\EqOrd}x\tau_{tt}\right)A^1\p_{A^1}-
 \tau_tA^0\p_{A^0},\\
 &\hat P(\chi)=\chi\partial_x-\chi_t\partial_{A^1},\quad \hat I(\phi)=\phi u\partial_u+\phi_t\partial_{A^0}.
\end{align*}
\end{corollary}

\begin{proof}
The proof follows immediately from the form of transformations constituting the equivalence group~$G^\sim$ 
and the fact that the algebra~$\mathfrak g^\sim$ consists 
of the infinitesimal generators of one-parameter subgroups of the equivalence group~$G^\sim$. 
We restrict $G^\sim$ to the continuous component by setting $\epsilon=1$ or $T_t>0$ 
and then successively assume that one of the parameter functions~$T$, $X^0$ and $U^1$ 
depend on a single continuous group parameter~$\varepsilon$ 
while the others take the same values as those for the identity transformation 
($t$, $0$, and $1$ for $T$, $X^0$ and~$U^1$, respectively), 
with the identity transformation corresponding to the parameter value~$\varepsilon=0$. 
The components of the infinitesimal generators of the form~$\hat Q=\tau\p_t+\xi\p_x+\eta\p_u+\psi^l\p_{A^l}$ 
are then obtained upon computing
\[
 \tau=\frac{\mathrm{d} \tilde t}{\mathrm{d} \varepsilon}\Big|_{\varepsilon=0},\quad   
 \xi=\frac{\mathrm{d} \tilde x}{\mathrm{d} \varepsilon}\Big|_{\varepsilon=0},\quad 
 \eta=\frac{\mathrm{d} \tilde u}{\mathrm{d} \varepsilon}\Big|_{\varepsilon=0},\quad 
 \psi^l=\frac{\mathrm{d} \tilde A^l}{\mathrm{d} \varepsilon}\Big|_{\varepsilon=0}, 
\]
which yields the spanning vector fields~$\hat D(\tau)$, $\hat P(\chi)$ and~$\hat I(\phi)$ associated to the parameter functions~$T$, $X^0$ and $U^1$, respectively,
\end{proof}

\begin{remark}\label{rem:GroupClassificationOfLinearHomogeneousGaugeEvolEqs}
The class~\eqref{eq:LinearHomogeneousGaugeEvolEqs} is a subclass of the class~\eqref{eq:GeneralLinearEvolutionEquation}. 
Each equation from the class~\eqref{eq:GeneralLinearEvolutionEquation} 
is $G^\sim_{\mbox{\tiny\eqref{eq:GeneralLinearEvolutionEquation}}}$-equivalent to 
an equation from the class~\eqref{eq:LinearHomogeneousGaugeEvolEqs}, 
and equations from the class~\eqref{eq:GeneralLinearEvolutionEquation} 
are $G^\sim_{\mbox{\tiny\eqref{eq:GeneralLinearEvolutionEquation}}}$-equivalent to each other 
if and only if their counterparts from the class~\eqref{eq:LinearHomogeneousGaugeEvolEqs} are $G^\sim$-equivalent. 
This is why the group classification of the class~\eqref{eq:GeneralLinearEvolutionEquation} 
reduces to that of the class~\eqref{eq:LinearHomogeneousGaugeEvolEqs}.
Moreover, in spite of not being normalized, 
the class~\eqref{eq:LinearHomogeneousGaugeEvolEqs} is more convenient for group classification 
than the class~\eqref{eq:GeneralLinearEvolutionEquation}. 
The class~\eqref{eq:LinearHomogeneousGaugeEvolEqs} has fewer number of arbitrary elements, 
and it is the uniform semi-normalization of the class~\eqref{eq:LinearHomogeneousGaugeEvolEqs} 
but not the normalization of the class~\eqref{eq:GeneralLinearEvolutionEquation} 
that allows to accurately neglect symmetry transformations of linear superposition of solutions
in the course of group classification.
\end{remark}

\section{Determining equations for Lie symmetries}\label{sec:DetEqsLieSymmetriesLinearEvolutionEqs}

The computation of the maximal Lie invariance group of an equation~$\mathcal L_A$ from the class~\eqref{eq:LinearHomogeneousGaugeEvolEqs} for a fixed tuple $A$ is readily realized using the Lie infinitesimal method. In particular, the generators of one-parameter point symmetry groups of~$\mathcal L_A$ are of the form $Q=\tau\p_t+\xi\p_x+\eta\p_u$
with the components~$\tau$, $\xi$ and~$\eta$ depending on $(t,x,u)$ and satisfy the infinitesimal invariance criterion, 
\[
 Q^{(\EqOrd)}(u_t-u_\EqOrd-A^lu_l)=0 \quad\text{for all solutions of~$\mathcal L_A$}.
\]
Here $Q^{(\EqOrd)}$ is the $\EqOrd$th prolongation of the vector field~$Q$, which reads
$
 Q^{(\EqOrd)}=Q+\sum_{0<|\alpha|\leqslant \EqOrd}\eta^\alpha\p_{u_\alpha}.
$
Recall that $\alpha=(\alpha_1,\alpha_2)$ is a multi-index, $|\alpha|=\alpha_1+\alpha_2$, and
$u_{\alpha}=\partial^{|\alpha|}u/\partial t^{\alpha_1}\partial x^{\alpha_2}$. 
The expressions for the components~$\eta^\alpha$ follow from the general prolongation formula~\cite{olve86Ay},
\[
\eta^\alpha=\mathrm D^\alpha\left(\eta-\tau u_t-\xi u_x\right)+\tau u_{\alpha+\delta_1}+\xi u_{\alpha+\delta_2},
\]
where $\mathrm D^\alpha=\mathrm D_t^{\alpha_1}\mathrm D_x^{\alpha_2}$, $\mathrm D_t=\partial_t+u_{\alpha+\delta_1}\partial_{u_\alpha}$ and $\mathrm D_x=\partial_x+u_{\alpha+\delta_2}\partial_{u_\alpha}$ are the total derivative operators with respect to~$t$ and~$x$, respectively, and $\delta_1=(1,0)$ and~$\delta_2=(0,1)$. The infinitesimal invariance criterion gives
\begin{equation}\label{eq:InfinitesimalInvarianceLinearEvolEqs}
\eta^{(1,0)}-\eta^{(0,\EqOrd)}-A^l\eta^{(0,l)}-(\tau A^l_t+\xi A^l_x)u_l=0 
\quad\mbox{if}\quad
u_t=u_\EqOrd-A^lu_l.
\end{equation}
Since we have shown above that the class~\eqref{eq:LinearHomogeneousGaugeEvolEqs} 
is uniformly semi-normalized with respect to linear superposition of solutions,
we can use all the restrictions on~$\tau$, $\xi$ and~$\eta$ 
derived in the course of the computation of the equivalence algebra 
also for the computation of the determining equations of Lie symmetries. 
In particular, we have
\[
 \tau=\tau(t),\quad \xi=\frac1\EqOrd\tau_t(t)x+\chi(t),\quad \eta=\phi(t)u+\eta^0(t,x),
\]
where $\eta^0(t,x)$ is a solution of Eq.~\eqref{eq:LinearHomogeneousGaugeEvolEqs}. 
With these restrictions on the components of infinitesimal generators of one-parameter Lie symmetry groups, 
the infinitesimal invariance condition~\eqref{eq:InfinitesimalInvarianceLinearEvolEqs} simplifies~to
\[
 \phi_tu+\eta^0_t-\xi_tu_x-\eta^0_\EqOrd-A^l\eta^0_l-\left(\tau A^l_t+\xi A^l_x+\frac{\EqOrd-l}{\EqOrd}\tau_tA^l\right)u_l=0.
\]
Splitting this equation with respect to the derivatives of~$u$ yields
\begin{subequations}\label{eq:DeterminingEqsLinearEvolEqs}
\begin{align}
 &\tau A^j_t+\left(\frac1\EqOrd\tau_tx+\chi\right)A^j_x+\frac{\EqOrd-j}{\EqOrd}\tau_tA^j=0,\quad j=2,\dots,\EqOrd-2,\label{eq:DeterminingEqsLinearEvolEqs1}\\
 &\tau A^1_t+\left(\frac1\EqOrd\tau_tx+\chi\right)A^1_x+\frac{\EqOrd-1}{\EqOrd}\tau_tA^1+\frac1\EqOrd\tau_{tt}x+\chi_t=0,\label{eq:DeterminingEqsLinearEvolEqs2}\\
 &\tau A^0_t+\left(\frac1\EqOrd\tau_tx+\chi\right)A^0_x+\tau_tA^0-\phi_t=0,\label{eq:DeterminingEqsLinearEvolEqs3}\\
 &\eta^0_t=\eta^0_\EqOrd+A^l\eta^0_l\label{eq:DeterminingEqsLinearEvolEqs4}.
\end{align}
\end{subequations}

The first three equations~\eqref{eq:DeterminingEqsLinearEvolEqs1}--\eqref{eq:DeterminingEqsLinearEvolEqs3} essentially depend on the parameter functions~$A^l\,$'s and are the \textit{classifying equations} for Lie symmetries of equations from the class~\eqref{eq:LinearHomogeneousGaugeEvolEqs}. The last equation~\eqref{eq:DeterminingEqsLinearEvolEqs4} is just a consequence of the linearity of the equation~$\mathcal L_A$ and hence is not a true classifying equation despite the fact that it depends on the arbitrary elements~$A$. We have thus proved the following assertion.

\begin{proposition}
The maximal Lie invariance algebra~$\mathfrak g_A$ of the equation~$\mathcal L_A$ from the class~\eqref{eq:LinearHomogeneousGaugeEvolEqs} consists of vector fields of the form~$Q=D(\tau)+P(\chi)+I(\chi)+Z(\eta^0)$ with
\begin{align*}
D(\tau)=\tau\p_t+\frac{1}{\EqOrd}\tau_tx\p_x,\quad
P(\chi)=\chi\p_x,\quad I(\phi)=\phi u\p_u,\quad Z(\eta^0)=\eta^0\p_u,
\end{align*}
where the parameter functions~$\tau$, $\chi$, $\phi$ and $\eta^0$ satisfy the classifying equations~\eqref{eq:DeterminingEqsLinearEvolEqs1}--\eqref{eq:DeterminingEqsLinearEvolEqs3} and $\eta^0$ runs through the solution set of~$\mathcal L_A$.
\end{proposition}

\begin{proposition}
The kernel invariance algebra~$\mathfrak g^\cap:=\bigcap_A\mathfrak g_A$ 
of equations from the class~\eqref{eq:LinearHomogeneousGaugeEvolEqs} 
is spanned by $I(1)$, $\mathfrak g^\cap=\langle I(1)\rangle$. 
\end{proposition}

\begin{proof}
In the derivation of the kernel of maximal Lie invariance algebras, 
the arbitrary elements~$A$ are assumed to be generic 
and hence one can split the determining equations~\eqref{eq:DeterminingEqsLinearEvolEqs} 
with respect to the arbitrary elements and their derivatives. 
This implies that $\tau=\xi=0$, $\phi=\const$ and $\eta^0=0$. 
\end{proof}

Let us now analyze the algebraic structure of the linear span
\[
\mathfrak g_\spanindex:=\langle D(\tau),P(\chi),I(\phi),Z(\zeta)\rangle,
\]
where the parameter functions~$\tau$, $\chi$ and~$\phi$ run through the set of smooth functions of~$t$, 
and the parameter function~$\zeta$ runs through the set of smooth functions of~$(t,x)$. 
Note that $\mathfrak g_\spanindex=\sum_A \mathfrak g_A$ 
since any vector field among~$D(\tau)$, $P(\chi)$, $P(1)+I(\phi)$ and $Z(\zeta)$ belongs to~$\mathfrak g_A$ for some~$A$.

The nonzero commutation relations between the vector fields spanning~$\mathfrak g_\spanindex$ are exhausted by
\begin{gather*}
[D(\tau),D(\check\tau)]=D(\tau\check\tau_t-\check\tau\tau_t),\quad 
[D(\tau),P(\chi)]=P\left(\tau\chi_t-\frac1\EqOrd\tau_t\chi\right),\quad
[D(\tau),I(\phi)]=I(\tau\phi_t),\\ 
[D(\tau),Z(\zeta)]=Z\left(\tau\zeta_t+\frac1\EqOrd\tau_tx\zeta_x\right),\quad
[P(\chi),Z(\zeta)]=Z(\chi\zeta_x),\quad 
[I(\phi),Z(\zeta)]=-Z(\phi\zeta).
\end{gather*}
These relations show that the span~$\mathfrak g_\spanindex$ is a Lie algebra with respect to the Lie bracket of vector fields. 
Moreover, it can be represented as a semi-direct sum,
\[
\mathfrak g_\spanindex=\mathfrak g^{\rm ess}_\spanindex\lsemioplus \mathfrak g^{\rm lin}_\spanindex, 
\quad \textup{where}\quad 
\mathfrak g^{\rm ess}_\spanindex=\langle D(\tau),P(\chi),I(\phi)\rangle 
\quad\textup{and}\quad 
\mathfrak g^{\rm lin}_\spanindex=\langle Z(\chi)\rangle
\]
are a subalgebra and an Abelian ideal of~$\mathfrak g_\spanindex$, respectively. 
The representation for the algebra~$\mathfrak g_\spanindex$ 
naturally translates into an analogous representation for each maximal Lie invariance algebra~$\mathfrak g_A$, 
\[
\mathfrak g_A=\mathfrak g^{\rm ess}_A\lsemioplus \mathfrak g^{\rm lin}_A
\quad\textup{with}\quad 
\mathfrak g^{\rm ess}_A=\mathfrak g_A\cap\mathfrak g^{\rm ess}_\spanindex
\quad\textup{and}\quad 
\mathfrak g^{\rm lin}_A=\mathfrak g_A\cap\mathfrak g^{\rm lin}_\spanindex=\langle Z(\eta^0),\eta^0\in\mathcal L_A\rangle,
\]
where, as before, $\eta^0$ is an arbitrary smooth solution of the equation~$\mathcal L_A$. 
In this decomposition, $\mathfrak g^{\rm ess}_A$ is a finite-dimensional subalgebra of~$\mathfrak g_A$ 
(see Lemma~\ref{lem:DimOfEssLieInvAlgebra} below), 
and the infinite-dimensional Abelian ideal~$\mathfrak g^{\rm lin}_A$ 
is spanned by the vector fields associated with linear superposition of solutions. 
Since the ideal~$\mathfrak g^{\rm lin}_A$ is a trivial a priori known part of~$\mathfrak g_A$, it is sufficient to focus on finding~$\mathfrak g^{\rm ess}_A$. This is why~$\mathfrak g^{\rm ess}_A$ is called the \emph{essential Lie invariance algebra} of the equation~$\mathcal L_A$.

Let~$\pi$ denote the projection of the joint space of equation variables and arbitrary elements of the class~\eqref{eq:LinearHomogeneousGaugeEvolEqs} to the space of equation variables only, 
$\pi(t,x,u,A)=(t,x,u)$. 
We clearly have $\mathfrak g^{\rm ess}_\spanindex=\pi_*\mathfrak g^\sim$, since the vector fields $\hat D(\tau)$, $\hat P(\chi)$ and~$\hat I(\phi)$ spanning~$\mathfrak g^\sim$ are mapped by~$\pi_*$ to the associated vector fields~$D(\tau)$, $P(\chi)$ and~$I(\phi)$ spanning~$\mathfrak g^{\rm ess}_\spanindex$. 
This is a manifestation of uniform semi-normalization with respect to linear superposition of solutions for the class~\eqref{eq:LinearHomogeneousGaugeEvolEqs}.
Furthermore, since the algebra~$\mathfrak g^{\rm ess}_\spanindex$ coincides with the set~$\pi_*\mathfrak g^\sim$ of the infinitesimal generators of one-parameter subgroups of the group~$\pi_*G^\sim$, the action of~$\pi_*G^\sim$ on~$\mathfrak g^{\rm ess}_\spanindex$ is well defined. The invariance of~$\mathfrak g^{\rm ess}_\spanindex$ and~$\mathfrak g^{\rm lin}_\spanindex$ under the action of~$\pi_*G^\sim$ implies that the action of~$G^\sim$ on equations from the class~\eqref{eq:LinearHomogeneousGaugeEvolEqs} also induces a well-defined action of~$\pi_*G^\sim$ on the essential Lie invariance algebras of these equations, which are subalgebras of~$\mathfrak g^{\rm ess}_\spanindex$. 
The kernel~$\mathfrak g^\cap$ is an ideal (more precisely, the center) of~$\mathfrak g^{\rm ess}_\spanindex$ and, hence, 
the ideal of~$\mathfrak g^{\rm ess}_A$ for each tuple of arbitrary elements~$A$.
The subalgebra~$\mathfrak s$ of~$\mathfrak g^{\rm ess}_\spanindex$ is called \emph{appropriate} if there exists a tuple~$A$ such that $\mathfrak s=\mathfrak g^{\rm ess}_A$.

This is why we have the following proposition; 
cf.\ the end of Section~\ref{sec:AlgebraicGroupClassificationMethod}.

\begin{proposition}
 The complete group classification of the class~\eqref{eq:LinearHomogeneousGaugeEvolEqs} of (1+1)-dimensional linear homogeneous evolution equations of order~$\EqOrd$ is accomplished by classifying all appropriate subalgebras of the algebra~$\mathfrak g^{\rm ess}_\spanindex$ with respect to the equivalence relation generated by the action of~$\pi_*G^\sim$.
 \looseness=-1
\end{proposition}

\section{Group classification}\label{sec:GroupClassificationLinearEvolutionEqs}

For the classification of appropriate subalgebras of the algebra~$\mathfrak g^{\rm ess}_\spanindex$, 
it is necessary to know the adjoint action of the transformations, $\mathcal T$, from~$\pi_*G^\sim$ 
on the vector fields, $Q$, from~$\mathfrak g^{\rm ess}_\spanindex$. 
Since we already know the finite form of transformations from~$\pi_*G^\sim$, 
it is convenient to compute the action of~$\mathcal T$ on~$Q$ directly by definition 
as the pushforward~$\mathcal T_*Q$ of~$Q$ by~$\mathcal T$~\cite{bihl11Dy,card11Ay},  
\[
 \mathcal T_* Q=Q(T)\p_{\tilde t}+Q(X)\p_{\tilde x}+Q(U)\p_{\tilde u},
\]
where the components of~$\mathcal T_* Q$ are expressed in terms of the transformed variables 
by substituting~$(t,x,u)=\mathcal T^{-1}(\tilde t,\tilde x,\tilde u)$ 
with the inverse transformation~$\mathcal T^{-1}$ of~$\mathcal T$. 
This method is especially suitable for infinite-dimensional Lie algebras. 

Consider the elementary transformations~$\mathcal D(T)$, $\mathcal P(X^0)$ and~$\mathcal I(U^1)$ from~$\pi_*G^\sim$, 
which are respectively obtained from~\eqref{eq:PointTransformationBetweenEvolEqs} with $\epsilon=1$ and $U^0=0$, 
where all but one of the parameter functions~$T$, $X^0$ and $U^1$ are set to trivial values, 
i.e., $t$ for $T$, zero for $X^0$ and one for $U^1$. 
If $\EqOrd$ is even, the group~$\pi_*G^\sim$ also contains the discrete transformation~$\mathcal X$ alternating the sign of~$x$, 
$\mathcal X$: $(\tilde t,\tilde x,\tilde u)=(t,-x,u)$. 
The nontrivial pushforwards of the generating vector fields of~$\mathfrak g^{\rm ess}_\spanindex$ 
by elementary transformations from~$\pi_*G^\sim$ are exhausted by
\begin{align}\label{eq:AdjointActionsEvolEqs}
\begin{split}
&\mathcal D_*(T)D(\tau)=\tilde D(T_t\tau),\quad 
 \mathcal D_*(T)P(\chi)=\tilde P(\sqrt[\EqOrd]{T_t}\chi),\quad 
 \mathcal D_*(T)I(\phi)=\tilde I(\phi),\\
&\mathcal P_*(X^0)D(\tau) =\tilde D(\tau)+\tilde P\left(\tau X^0_t-\frac1\EqOrd\tau_tX^0\right),\quad 
 \mathcal I_*(U^1)D(\tau)=\tilde D(\tau)+\tilde I\left(\tau \frac{U^1_t}{U^1}\right),\\ 
&\EqOrd\in2\mathbb N\colon\quad\mathcal X_*P(\chi)=\tilde P(-\chi).
\end{split}
\end{align}
The tildes over the operators on the right-hand side indicates 
that these vector fields are expressed in terms of the transformed variables, 
and we should substitute for~$t$, $t=T^{-1}(\tilde t)$, where $T^{-1}$ is the inverse function of~$T$.

We now derive the upper bound for the dimension of essential Lie invariance algebras 
for equations from the class~\eqref{eq:LinearHomogeneousGaugeEvolEqs}.

\begin{lemma}\label{lem:DimOfEssLieInvAlgebra}
$\dim\mathfrak g^{\rm ess}_A\leqslant 4$
for any tuple of arbitrary elements~$A$.
\end{lemma}

\begin{proof}
The proof is analogous to the one of Lemma~18 in~\cite{kuru16a}. 
Let the equation~$\mathcal L_A$ be defined on the domain $\Omega_t\times \Omega_x$, 
where $\Omega_t\subseteq\mathbb R$ and $\Omega_x\subseteq\mathbb R$ 
are open intervals on the $t$- and $x$-axes, respectively. 
We evaluate the classifying equations~\eqref{eq:DeterminingEqsLinearEvolEqs2} and~\eqref{eq:DeterminingEqsLinearEvolEqs3} 
at two distinct points $x_0$ and $x_1$ from $\Omega_x$ and vary $t$. 
This yields
 \begin{align*}
  &\frac1\EqOrd\tau_{tt}x_1+\chi_t=-R^1,\quad \frac1\EqOrd\tau_{tt}x_0+\chi_t=-R^2,\quad \phi_t=R^3,
 \end{align*} 
where $R^1$ and $R^2$ are obtained upon substituting~$x_1$ and~$x_0$ 
into the part of~\eqref{eq:DeterminingEqsLinearEvolEqs2} involving~$A^1$, respectively, 
and $R_3$ is obtained by substituting either $x_0$ or $x_1$ 
into the part of~\eqref{eq:DeterminingEqsLinearEvolEqs3} involving~$A^0$. 
Since the points $x_0$ and $x_1$ are distinct, the above system can be brought into a system of linear ordinary differential equations in the canonical form, 
\[
 \tau_{tt}=\dots,\quad \chi_t=\dots,\quad \phi_t=\dots,
\]
where the precise form of the respective right hand sides of the resolved system is inessential for the current investigation. 
The solution space of this linear system for~$\tau$, $\chi$ and $\phi$ is obviously four-dimensional.
Since more such equations as the above ones may be derived from the classifying equations, 
we have $\dim\mathfrak g^{\rm ess}_A\leqslant 4$.
\end{proof}

Having established the maximum dimension of appropriate subalgebras of~$\mathfrak g^{\rm ess}_\spanindex$, we now proceed to restrict their form. 
This is done in the following series of lemmas, which are the analogous results to those given in~\cite{kuru16a} for the class of linear Schr\"odinger equations. 
It is convenient to introduce the three integer numbers~$k_0$, $k_1$ and~$k_2$ below, which characterize the dimensions of relevant subalgebras of~$\mathfrak g^{\rm ess}_\spanindex$.  
\begin{lemma}\label{lem:EssentialAlgebraContainingKernel}
$\mathfrak g^{\rm ess}_A\cap \langle I(\phi)\rangle=\mathfrak g^\cap$ 
and thus $k_0:=\dim\mathfrak g^{\rm ess}_A\cap\langle I(\phi)\rangle=1$
for any tuple of arbitrary elements~$A$. 
\end{lemma}

\begin{proof}
The intersection $\mathfrak g^{\rm ess}_A\cap\langle I(\phi)\rangle$ is included in~$\mathfrak g^\cap$ since the classifying equations~\eqref{eq:DeterminingEqsLinearEvolEqs2} and~\eqref{eq:DeterminingEqsLinearEvolEqs3} for $\tau=\chi=0$ imply~$\phi_t=0$. 
On the other hand, the kernel invariance algebra~$\mathfrak g^\cap$ is contained in~$\mathfrak g^{\rm ess}_A$ for any~$A$ 
and thus we have $\mathfrak g^\cap\subset\mathfrak g^{\rm ess}_A\cap\langle I(\phi)\rangle$. 
These two inclusions jointly prove the lemma.
\end{proof}

Lemma~\ref{lem:EssentialAlgebraContainingKernel} implies that 
$\mathfrak g^{\rm ess}_\spanindex\supsetneq\bigcup_A\mathfrak g^{\rm ess}_A$. 
Moreover, $\mathfrak g^{\rm ess}_\spanindex\setminus\bigcup_A\mathfrak g^{\rm ess}_A=\{I(\phi)\mid\phi\ne\const\}$ 
since the classifying equations~\eqref{eq:DeterminingEqsLinearEvolEqs} imply 
that each vector field from the complement of $\{I(\phi)\mid\phi\ne\const\}$ with respect to $\mathfrak g^{\rm ess}_\spanindex$ 
belongs to $\bigcup_A\mathfrak g^{\rm ess}_A$.

\begin{lemma}
$k_1:=\dim\mathfrak g^{\rm ess}_A\cap \langle P(\chi), I(\phi)\rangle-1 \in\{0,1\}$
for any tuple of arbitrary elements~$A$.
\end{lemma}

\begin{proof}
Denote $\mathfrak a_A:=\mathfrak g^{\rm ess}_A\cap \langle P(\chi), I(\phi)\rangle$. 
Since $\mathfrak a_A\supset\mathfrak g^\cap$, then $\dim\mathfrak a_A\geqslant1$. 
Similar to the proof of Lemma~\ref{lem:DimOfEssLieInvAlgebra}, the classifying equations~\eqref{eq:DeterminingEqsLinearEvolEqs2} and~\eqref{eq:DeterminingEqsLinearEvolEqs3} 
imply, at least, a system of two linear ordinary differential equations for the parameter-functions $\chi$ and~$\phi$ depending on~$t$, $\chi_t=\dots,\quad \phi_t=\dots,$ 
whose solution space is two-dimensional. Therefore, $\dim\mathfrak a_A\leqslant2$.
\end{proof}

The projection~$\varpi$ on the space of~$t$ defines the mapping~$\varpi_*$ on $\mathfrak g^{\rm ess}_\spanindex$, $D(\tau)+P(\chi)+I(\phi)\mapsto\tau\p_t$, and thus 
$\varpi_*\mathfrak g^{\rm ess}_\spanindex=\langle\tau\p_t\rangle$, 
where $\tau$ runs through the set of smooth functions of~$t$. 
The pushforward~$\varpi_* G^\sim$ of~$G^\sim$ by the projection~$\varpi$ is also well defined. 

\begin{lemma}\label{lem:OnOperatorsInvolvingTau}
The projection~$\varpi_*\mathfrak g^{\rm ess}_A$ is a Lie algebra for any tuple of arbitrary elements~$A$ and $k_2:=\dim\varpi_*\mathfrak g^{\rm ess}_A\leqslant 2$. Moreover, $\varpi_*\mathfrak g^{\rm ess}_A\in\{0,\langle\p_t\rangle,\langle\p_t,t\p_t\rangle\}\bmod\varpi_* G^\sim$.
\end{lemma}

\begin{proof}
We first show that~$\varpi_*\mathfrak g^{\rm ess}_A$ is indeed a Lie algebra. 
Let there be given~$\tau^i\p_t\in\varpi_*\mathfrak g^{\rm ess}_A$, $i=1,2$. 
There exist~$Q_i\in\mathfrak g^{\rm ess}_A$ such that~$\varpi_*Q_i=\tau^i\p_t$. 
For any constants~$c_1$ and~$c_2$ it follows that~$c_1Q_1+c_2Q_2\in\mathfrak g^{\rm ess}_A$. 
Thus, $c_1\tau^1\p_t+c_2\tau^2\p_t=\varpi_*(c_1Q_1+c_2Q_2)\in\varpi_*\mathfrak g^{\rm ess}_A$, 
proving that~$\varpi_*\mathfrak g^{\rm ess}_A$ is a linear space. 
Since 
\[[\tau^1\p_t,\tau^2\p_t]=(\tau^1\tau^2_t-\tau^2\tau^1_t)\p_t=\varpi_*[Q_1,Q_2]\in\varpi_*\mathfrak g^{\rm ess}_A,\] 
this space is closed under the Lie bracket of vector fields and thus it is a Lie algebra with 
$\dim\varpi_*\mathfrak g^{\rm ess}_A\leqslant\dim\mathfrak g^{\rm ess}_A-\dim\mathfrak g^{\cap}\leqslant 3$. 

Moreover, the pushforward~$\varpi_* G^\sim$ of~$G^\sim$ by the projection~$\varpi$ coincides with the (pseudo)group of local diffeomorphisms in the space of~$t$.
This is why we can use the Lie theorem stating that the maximum dimension of finite-dimensional Lie algebras of vector fields on the complex (resp.\ real) line is three, 
and, up to local diffeomorphisms of the line, 
these algebras are exhausted by~$\{0\}$, $\langle\p_t\rangle$, $\langle\p_t,t\p_t\rangle$ and~$\langle\p_t,t\p_t,t^2\p_t\rangle$. 
 
We prove by contradiction that the last algebra cannot serve as $\varpi_*\mathfrak g^{\rm ess}_A$ for some~$A$. 
Suppose that this is not the case for a tuple~$A$. 
Then the algebra~$\mathfrak g^{\rm ess}_A$ coincides with the span $\langle Q_0,Q_1,Q_2,Q_3\rangle$, 
where the vector field $Q_0=I(1)$ spans the kernel algebra~$\mathfrak g^\cap$, and
\[
 Q_i=D(t^{i-1})+P(\chi^i)+I(\phi^i),\quad i=1,2,3,
\]
with smooth functions~$\chi^i$ and~$\phi^i$ of~$t$. 
Using the adjoint actions $\mathcal P_*(-\int\!\chi^1\mathrm{d}t)$ and $\mathcal I_*(e^{-\int\!\phi^1\mathrm{d}t})$,
cf.\ \eqref{eq:AdjointActionsEvolEqs}, we can set~$\chi^1=\phi^1=0$. 
In what follows $a$'s, $b$'s and~$c$'s denote constants. 
Commuting~$Q_1$ and~$Q_2$ gives 
\[
[Q_1,Q_2]=D(1)+P(\chi^2_t)+I(\phi^2_t)=a_0Q_0+a_1Q_1+a_2Q_2+a_3Q_3,
\] 
which implies that~$a_1=1$, $a_2=a_3=0$, and hence~$\chi^2_t=0$ and $\phi^2_t=a_0$, i.e.\ $\phi^2=a_0t+a_4$. 
The adjoint actions $\mathcal P_*(\EqOrd\chi^2)$ and $\mathcal I_*(e^{-a_0t})$ 
and the recombination $Q_1+a_0Q_0\to Q_1$, $Q_2-a_4Q_0\to Q_2$ allow setting $\chi^2=\phi^2=0$ without modifying the form of~$Q_1$.

Next, we analyze the commutation relations of~$Q_1$ and~$Q_2$ with~$Q_3$,
\begin{gather*}
[Q_1,Q_3]=2D(t)+P(\chi^3_t)+I(\phi^3_t)=b_0Q_0+b_1Q_1+b_2Q_2+b_3Q_3,\\ 
[Q_2,Q_3]=D(t^2)+P(t\chi^3_t-\tfrac1\EqOrd\chi^3)+I(t\phi^3_t)=c_0Q_0+c_1Q_1+c_2Q_2+c_3Q_3. 
\end{gather*}
They directly imply that $b_1=b_3=0$, $b_2=2$, $c_1=c_2=0$, $c_3=1$ 
and thus~$\chi^3_t=0$, $t\chi^3_t=(1+\tfrac1\EqOrd)\chi^3$, 
i.e., $\chi^3=0$. 
Since the classifying equation~\eqref{eq:DeterminingEqsLinearEvolEqs2} is only of second order in~$\tau$, 
the three vector fields~$Q_1=D(1)$, $Q_2=D(t)$ and~$Q_3=D(t^2)+I(\phi^3)$, whose parameters $\chi$'s are zero, cannot simultaneously satisfy this equation. 
The derived contradiction completes the proof.
\end{proof}

The above lemmas jointly imply that any appropriate subalgebra of~$\mathfrak g^{\rm ess}_\spanindex$ is spanned by
\begin{enumerate}\itemsep=-0.5ex
 \item the basis vector field~$Q_0=I(1)$ of the kernel~$\mathfrak g^\cap$, $k_0=1$;
 \item $k_1\leqslant 1$ vector fields of the form~$Q_i=P(\chi^i)+I(\phi^i)$, 
       where $i=1,\dots,k_1$, and $\chi^1\ne0$ if~$k_1=1$;
 \item $k_2\leqslant2$ vector fields of the form~$Q_i=D(\tau^i)+P(\chi^i)+I(\phi^i)$, 
       where $i=k_1+1,\dots,k_1+k_2$, and~$\tau^{k_1+1}$, \dots, $\tau^{k_1+k_2}$ are linearly independent.
\end{enumerate}

Moreover, we have that~$\dim\mathfrak g^{\rm ess}_A=k_0+k_1+k_2=1+k_1+k_2\leqslant4$.

\begin{theorem}
 A complete list of~$G^\sim$-inequivalent (and, therefore, $\mathcal G^\sim$-inequivalent) Lie symmetry extensions in the class~\eqref{eq:LinearHomogeneousGaugeEvolEqs} is exhausted by the cases given in Table~\ref{tab:CompleteGroupClassificationLinearEvolEqs}.
\end{theorem}

\begin{table}[!ht]
\begin{center}
\caption{Complete group classification of the class~\eqref{eq:LinearHomogeneousGaugeEvolEqs}.
\label{tab:CompleteGroupClassificationLinearEvolEqs}}
\def\arraystretch{1.45}
\begin{tabular}{|c|c|c|c|l|}
\hline
no. & $k_1$ & $k_2$ & $A^l$                             & \hfil Basis of~$\mathfrak g^{\rm ess}_A$\\
\hline
0 & 0 & 0 & $A^l=A^l(t,x)$                              & $I(1)$\\
1 & 0 & 1 & $A^l=A^l(x)$                                & $I(1), D(1)$\\
2 & 0 & 2 & $A^l=c_lx^{l-\EqOrd}$                            & $I(1), D(1), D(t)$\\
3 & 1 & 0 & $A^j=A^j(t),\quad A^1=0,\quad A^0=f(t)x$    & $I(1), P(1)+I\left(\int\!f\,\mathrm{d}t\right)$\\
4a& 1 & 1 & $A^j=c_j,\quad A^1=0,   \quad A^0=\sigma x$ & $I(1), D(1), P(1)+\sigma I(t)$\\
4b& 1 & 1 & $A^j=c_j,\quad A^1=-x,  \quad A^0=\sigma x$ & $I(1), D(1), P(e^{t})+\sigma I(e^t)$\\
5 & 1 & 2 & $A^l=0$                                     & $I(1), D(1), D(t), P(1)$\\
\hline
\end{tabular}
\end{center}
{\footnotesize
Here $l$ runs from~$0$ to $\EqOrd-2$, $j$ runs from $2$ to $\EqOrd-2$, $c$'s and~$\sigma$ are constants,
and all functional parameters are smooth functions of their arguments.
The presented vector fields span $\mathfrak g^{\rm ess}_A$ 
if the equation~$\mathcal L_A$ is not $G^\sim$-equivalent to an equation with a wider essential Lie invariance algebra. 
In particular, at least one of the constant parameters in~$A$'s should be nonzero in Cases~2, 4a and~4b for this.
The parameter functions in Case~3 should be neither simultaneously constant 
nor, up to shifts of~$t$, of the form given in Remark~\ref{rem:AlternativeFormOfCase4b} below. 
In Case~1, the tuple~$A$ should be neither of the form 
$(c_0(x+a)^{-\EqOrd}+b_0,c_1(x+a)^{1-\EqOrd}+b_2x+b_1,c_j(x+a)^{j-\EqOrd},j=2,\dots,\EqOrd-2)$ nor 
 of the form $(b_0x+c_0,b_1x+c_1,c_j,j=2,\dots,\EqOrd-2)$ for some constants~$a$, $b$'s and~$c$'s. 
See also the proof for gauging of parameters.}
\end{table}

\begin{proof}
The classification cases follow from studying the different values of~$k_1$ and~$k_2$, 
and the case enumeration in the proof coincides with that in Table~\ref{tab:CompleteGroupClassificationLinearEvolEqs}. 
Choosing a general representation for basis vector fields of an appropriate subalgebra~$\mathfrak s$ of~$\mathfrak g^{\rm ess}_\spanindex$
for each possible value of~$(k_1,k_2)$, 
we simplify them using the adjoint actions of equivalence transformations given in~\eqref{eq:AdjointActionsEvolEqs} 
and linear recombination of these basis vector fields. 
For each case, we also include~$I(1)$ in the basis as the basis element of the kernel invariance algebra~$\mathfrak g^\cap$.
In view of Lemma~\ref{lem:EssentialAlgebraContainingKernel}, 
any vector field $D(\tau)+P(\chi)+I(\phi)$ from $\mathfrak g^{\rm ess}_A\setminus\mathfrak g^\cap$ has $(\tau,\chi)\ne(0,0)$ and, hence,
can be reduced, up to $\pi_*G^\sim$-equivalence, to the form~$D(1)$ or~$P(1)+I(\phi)$ if~$\tau\ne0$ or~$\tau=0$ and~$\chi\ne0$, respectively.
If $k_1+k_2>1$, after a preliminary simplification of basis vector fields we should take into account the fact 
that the corresponding span~$\mathfrak s$ is closed with respect to the Lie brackets of vector fields. 
Then we make further simplifications. 
Evaluating the classifying equations \eqref{eq:DeterminingEqsLinearEvolEqs1}--\eqref{eq:DeterminingEqsLinearEvolEqs3} 
at the simplified basis vector fields results in a system of differential equations for the arbitrary elements~$A$. 
This system should be integrated up to equivalence transformations 
whose pushforwards by~$\pi$ preserve, up to linear recombination, the form of the simplified basis vector fields
while possibly changing parameters that may still remain in these vector fields. 
For convenience, the group of such transformations is denoted by~$G^\sim_{\mathfrak s}$. 

Again, in what follows $a$'s, $b$'s and~$c$'s denote constants, $l=0,\dots,\EqOrd-2$, $m=1,\dots,\EqOrd-2$, $j=2,\dots,\EqOrd-2$. 

\medskip
 
\noindent\textbf{0.}\ $k_1=k_2=0$. 
This case corresponds to the general equation from the class~\eqref{eq:LinearHomogeneousGaugeEvolEqs}, 
for which there is no extension of the kernel invariance algebra, $\mathfrak s=\mathfrak g^\cap$, and $G^\sim_{\mathfrak s}=G^\sim$.
 
 \medskip
 
\noindent\textbf{1.}\ $k_1=0$, $k_2=1$. 
Up to~$\pi_*G^\sim$-equivalence, the algebra~$\mathfrak s$ is spanned by the vector fields~$D(1)$ and~$I(1)$. Then the classifying equations \eqref{eq:DeterminingEqsLinearEvolEqs1}--\eqref{eq:DeterminingEqsLinearEvolEqs3} imply that~$A^l_t=0$. 
The constraints singling out the subgroup~$G^\sim_{\mathfrak s}$ from~$G^\sim$ are $T_{tt}=0$, $X^0_t=0$ and $(U^1_t/U^1)_t=0$.

\medskip
 
\noindent\textbf{2.}\ $k_1=0$, $k_2=2$. 
Similarly to the proof of Lemma~\ref{lem:OnOperatorsInvolvingTau}, we can assume up to $\pi_*G^\sim$-equivalence that 
$\mathfrak s=\langle I(1),D(1),D(t)\rangle$. 
Evaluating the classifying equations \eqref{eq:DeterminingEqsLinearEvolEqs1}--\eqref{eq:DeterminingEqsLinearEvolEqs3} 
at the vector fields $D(1)$ and~$D(t)$ gives 
the system $A^l_t=0$, $xA^l_x+(\EqOrd-l)A^l=0$, 
whose general solution is $A^l=c_lx^{l-\EqOrd}$, where $c_l$'s are arbitrary constants. 
The subgroup~$G^\sim_{\mathfrak s}$ is singled out from~$G^\sim$ by the constraints $T_{tt}=0$, $X^0=U^1=0$ 
and thus $\pi_*G^\sim_{\mathfrak s}$ coincides with~the essential point symmetry group of~$\mathcal L_A$. 
This is why no simplification of~$A$ is possible here. 
 
\medskip
 
\noindent\textbf{3.}\ $k_1=1$, $k_2=0$. 
Up to $\pi_*G^\sim$-equivalence, the extension to the kernel algebra is given by a single operator of the form~$P(1)+I(\phi)$ with no restrictions placed on~$\phi$. 
The classifying equations \eqref{eq:DeterminingEqsLinearEvolEqs1}--\eqref{eq:DeterminingEqsLinearEvolEqs3} imply 
$A^m_x=0$ and $A^0_x=\phi_t$. 
The group~$G^\sim_{\mathfrak s}$ consists of the equivalence transformations with $T_{tt}=0$.
Hence the general solution of the system for~$A$ can be simplified by setting $A^1=0$ and $A^0=f(t)x$, where $f=\phi_t$. 
Equivalence transformations preserving this form of~$A^1$ and~$A^0$ 
are exhausted by the ones with $T_{tt}=0$, $X^0_t=0$ and $U^1_t/U^1=fX^0/X^1$.

\medskip

\noindent\textbf{4.}\ $k_1=1$, $k_2=1$. 
Up to $\pi_*G^\sim$-equivalence, $\mathfrak s=\langle Q_0,Q_1,Q_2\rangle$, 
where $Q_0=I(1)$, $Q_1=P(\chi)+I(\phi)$ with $\chi\ne0$ and $Q_2=D(1)$. 
Commuting $Q_2$ and~$Q_1$ gives
\[
 [Q_2,Q_1]=P(\chi_t)+I(\phi_t)=a_0Q_0+a_1Q_1+a_2Q_2,
\] 
which implies that~$a_2=0$. Then, two subcases are possible depending on the value of~$a_1$. 

Firstly, $a_1=0$ and hence $\chi_t=0$ and $\phi_t=a_0$, implying~$\phi=a_0t+a_3$. 
We can set $\chi=1$ and $a_3=0$ by recombining $(1/\chi)(Q_1-a_3Q_0)\to Q^1$, 
i.e., we obtain the new $Q_1=P(1)+\sigma I(t)$, where $\sigma:=a_0/\chi$.
The system derived by evaluating the classifying equations \eqref{eq:DeterminingEqsLinearEvolEqs1}--\eqref{eq:DeterminingEqsLinearEvolEqs3} 
at $Q^1$ and $Q^2$ consists of the equations $A^l_t=0$, $A^m_x=0$, $A^0_x=\sigma$, 
and its general solution is $A^m=c_m$, $A^0=\sigma x+c_0$.
The subgroup~$G^\sim_{\mathfrak s}$ is singled out from~$G^\sim$ by the constraints $T_{tt}=0$, $X^0_{tt}=0$ and $(U^1_t/U^1)_t=\sigma X^0_t$. 
Hence $c_1=c_0=0\bmod G^\sim$. This yields Case~4a. 
One of the constants~$c_j$'s or~$\sigma$, if it is nonzero, can always be scaled to $\pm1$, 
and further gauging to~1 is possible for a~$c_j$ if $\EqOrd-j$ is odd or for~$\sigma$ if $\EqOrd$ is even.

If~$a_1\ne0$, then we can scale~$a_1$ to $1$ by scaling of~$t$ and set~$a_0$ to zero upon linearly combining~$Q_1$ with~$Q_0$. 
Then the above commutation relation implies that~$\chi_t=\chi$ and~$\phi_t=\phi$, 
yielding, up to scaling of~$Q_1$, the solution~$\chi=e^t$ and~$\phi=\sigma e^t$, 
where $\sigma=\const$, i.e., we get $Q_1=P(e^t)+\sigma I(e^t)$. 
The invariance of~$\mathcal L_A$ with respect to~$Q_1$ and~$Q_2$ requires 
for~$A$ to satisfy the system $A^l_t=0$, $A^j_x=0$, $A^1_x=-1$, $A^0_x=\sigma$ 
and hence to be of the form $A^j=c_j$, $A^1=-x+c_1$, $A^0=\sigma x+c_0$.
The constraints for equivalence transformations constituting the subgroup~$G^\sim_{\mathfrak s}$ 
are $T_t=1$, $X^0_{tt}=X^0_t$ and $(U^1_t/U^1)_t=\sigma X^0_t$. 
Since $c_1=c_0=0\bmod G^\sim$, we obtain Case~4b. 
If $\EqOrd$ is even, then one of the constants~$c_j$'s with odd~$j$ 
or~$\sigma$, if this constant is nonzero, can be assumed positive 
due to alternating the sign of~$x$. 
 
\medskip

\noindent\textbf{5.}\ $k_1=1$, $k_2=2$. 
After a preliminary simplification, the vector fields spanning~$\mathfrak s$ take the form 
$Q_0=I(1)$, $Q_1=P(\chi^1)+I(\phi^1)$, $Q_2=D(1)$ and $Q_3=D(t)+P(\chi^3)+I(\phi^3)$, where $\chi^1\ne0$. 
The nonzero commutation relations between the basis elements are exhausted by 
\begin{gather*}
[Q_2,Q_1]=P(\chi^1_t)+I(\phi^1_t)=a_0Q_0+a_1Q_1+a_2Q_2+a_3Q_3,\\
[Q_3,Q_1]=P(t\chi^1_t-\tfrac1\EqOrd\chi^1)+I(t\phi^1_t)=b_0Q_0+b_1Q_1+b_2Q_2+b_3Q_3,\\
[Q_2,Q_3]=D(1)+P(\chi^3_t)+I(\phi^3_t)=c_0Q_0+c_1Q_1+c_2Q_2+c_3Q_3,
\end{gather*}
which immediately imply~$a_2=a_3=0$, $b_2=b_3=0$, $c_2=1$, $c_3=0$, and then give the equations
\begin{gather*}
\chi^1_t=a_1\chi^1, \quad \phi^1_t=a_1\phi^1+a_0,\\
t\chi^1_t-\tfrac1\EqOrd\chi^1=b_1\chi^1, \quad t\phi^1_t=b_1\phi^1+b_0,\\
\chi^3_t=c_1\chi^1, \quad \phi^3_t=c_1\phi^1+c_0.
\end{gather*}
The compatibility of the first four equations with respect to~$\chi^1$ and~$\phi^1$
requires $a_1=0$, $b_1=-\frac1\EqOrd$ and $a_0=0$. 
Hence the parameters~$\chi^1$ and~$\phi^1$ are constants. 
We can set $\chi^1=1$ and $\phi^1=0$ by recombining $(1/\chi^1)(Q_1-\phi^1Q_0)\to Q^1$ and get $Q_1=P(1)$.

Due to the previous simplification of~$Q_1$, 
the parameters~$\chi^3$ and~$\phi^3$ satisfy the equations $\chi^3_t=c_1$ and $\phi^3_t=c_0$, 
which can be integrated to~$\chi^3=c_1t$ and~$\phi^3=c_0t$ 
with the integration constants set to zero at once after recombining~$Q_3$ with~$Q_1$ and~$Q_0$, respectively. 
Applying the adjoint actions~$\mathcal P_*(\tilde c_1t)$ with~$\tilde c_1=-\EqOrd c_1/(\EqOrd-1)$ and~$I_*(e^{-c_0t})$, 
we can set~$c_1=c_0=0$ and hence obtain~$Q_3=D(t)$. 
The form of~$Q_2$ is restored by recombining~$Q_2$ with~$Q_1$ and~$Q_0$.
The evaluation of the classifying equations~\eqref{eq:DeterminingEqsLinearEvolEqs1}--\eqref{eq:DeterminingEqsLinearEvolEqs3} 
at $Q_1$, $Q_2$ and $Q_3$ gives the system for~$A$ with the only zero solution, $A^l=0$. 
\end{proof}

\begin{corollary}
 A linear $\EqOrd$th order evolution equation is reduced to the simplest form $u_t=u_\EqOrd$ by a point transformation 
 if and only if its essential Lie invariance algebra is four-dimensional. 
\end{corollary}

\begin{remark}
Case~1 (resp.\ Case~3) of Table~\ref{tab:CompleteGroupClassificationLinearEvolEqs} contains 
equations equivalent to each other. 
At the same time, the corresponding equivalence transformations 
cannot properly be used for gauging of arbitrary elements. 
\end{remark}

\begin{remark}\label{rem:AlternativeFormOfCase4b}
 The form of the algebra~$\mathfrak g^{\rm ess}_A$ in Case~4b of Table~\ref{tab:CompleteGroupClassificationLinearEvolEqs} is not common. 
 The equivalence transformation~\eqref{eq:PointTransformationBetweenEvolEqs}--\eqref{eq:PointTransformationBetweenEvolEqsAE1} 
 with $T=-e^{-\EqOrd t}$, $X^0=0$ and $U^1=1$ reduces Case~4b to the case 
\begin{center}%
\def\arraystretch{1.5}\tabcolsep=1ex
\begin{tabular}{|c|c|c|c|l|}
\hline
$\tilde{\text{4b}}$ & 1 & 1 & $A^j=\tilde c_j|t|^{j/\EqOrd-1},\ A^1=0,\ A^0=\tilde\sigma|t|^{-1/\EqOrd-1} x$ & $I(1), D(t), P(1)-\EqOrd\tilde\sigma\sgn(t)I(|t|^{-1/\EqOrd})$\\
\hline
\end{tabular}
\end{center}
\noindent
where $\tilde c_j=\EqOrd^{j/\EqOrd-1}c_j$, $\tilde\sigma=\EqOrd^{-1/\EqOrd-1}\sigma$,
and we omit tildes of the variables and the arbitrary elements. 
Case~$\tilde{\text{4b}}$ enhances and extends Case~5 of \cite[Table~1]{gung04a} with $\EqOrd=3$ and Cases~5 and~7 of \cite[Table~1]{huan12a} with $\EqOrd=4$ 
to arbitrary~$\EqOrd$.
While Case~$\tilde{\text{4b}}$ is more consistent with the tree of Lie symmetry extensions in Table~\ref{tab:CompleteGroupClassificationLinearEvolEqs} 
than Case~4b, and the possibility of the further Lie symmetry extension to Case~5 is more obvious for it, 
the parameter functions~$A$'s become time-dependent with fractional negative powers of~$|t|$. 
This is why the presented form of Case~4b is simpler and thus seems preferable.  
\end{remark}

\begin{remark}\label{rem:ComparisonWithThreeFourOrderCase}
It is worth comparing the entire group classification results 
given in \cite[Table~1]{gung04a} and \cite[Table~1]{huan12a}
for the classes of (1+1)-dimensional third- and fourth-order linear evolution equations
with Table~\ref{tab:CompleteGroupClassificationLinearEvolEqs} for $\EqOrd=3$ and $\EqOrd=4$, respectively.
These results should in fact be identical. 

Thus, Cases~1--6 of~\cite[Table~1]{gung04a} correspond to Cases~1, 3, 2, 4a, $\tilde{\text{4b}}$ and~5 
of Table~\ref{tab:CompleteGroupClassificationLinearEvolEqs}, respectively. 
A minor inconvenience of \cite[Table~1]{gung04a} is 
that the correct condition for~$f$ in Case~2 is $\ddot f(t)\ne0$ but not $\dot f(t)\ne0$.
A more essential inconvenience is that the parameter~$a$ in Case~6 of \cite[Table~1]{gung04a} can be set to zero 
by the equivalence transformation with~$T=t$, $X^0=at$ and~$U^1=1$, 
which leads to the simplest equation~$u_t=u_3$. 

In~\cite{huan12a}, some classification cases were redundantly split into pairs of their subcases. 
As a result, Cases~1, 8, (2,\,3), (4,\,6), (5,\,7) and~10 of~\cite[Table~1]{huan12a} 
correspond to our Cases~1, 2, 3, 4a, $\tilde{\text{4b}}$ and~5, respectively. 
The major weakness  of~\cite[Table~1]{huan12a} is Case~9, 
which can be obviously simplified by $\mathcal P(b/a)$ to the case with 
$A=C=0$, $B=a_1x^{-3}+cx$ (in the notation of~\cite{huan12a}), where $a_1=-cb^4a^{-4}$, 
and the essential Lie invariance algebra $\langle I(1),D(1),D(e^{-4ct})\rangle$.
Moreover, the simplified Case~9 is mapped by the equivalence transformation $\mathcal D(e^{4ct})$  
to a subcase of Case~8 of the same table
and, therefore, Case~9 should be excluded from the classification~list.  
\end{remark}

\begin{remark}
Any (1+1)-dimensional second-order linear evolution equation is similar to an equation of the same kind 
whose essential Lie invariance algebra is a subalgebra of the essential Lie invariance algebra of 
the simplest equation $u_t=u_2$, which is the heat equation. 
There is no analogue of this property for (1+1)-dimensional $\EqOrd$th order linear evolution equations with $\EqOrd>2$ 
although the essential Lie invariance algebra of the simplest equation $u_t=u_\EqOrd$ is still of maximum dimension among such equations.
\end{remark}

\section{Exact solutions}\label{sec:ExactSolutionsLinEvolEqs}

A standard procedure for finding exact solutions of a given system of partial differential equations 
with nonzero maximal Lie invariance algebra~$\mathfrak a$
is to carry out Lie reductions. 
Here one firstly classifies inequivalent subalgebras, of appropriate dimensions, of the algebra~$\mathfrak a$ and then constructs solution ansatzes using invariants of the listed inequivalent subalgebras. Substituting these ansatzes into the original system then leads to reduced systems with fewer independent variables, which are in general easier to solve than the original system. If a solution of a reduced system is constructed, then substituting it into the associated ansatz gives a particular solution of the original system. 
See~\cite{olve86Ay,ovsi82Ay} for more details.

For classes of differential equations, it seems natural to separately carry out the whole Lie reduction procedure for each inequivalent case of Lie symmetry extension. However, for classes that are normalized it is in fact sufficient to classify only low-dimensional subalgebras of the equivalence algebra up to $G^\sim$-equivalence and carry out Lie reductions using the invariants of the projections of the associated basis vector fields to the space of equation variables~\cite{poch16a}. For classes that are not normalized, this strategy does not yield complete reduction results.

In the present case, the class~\eqref{eq:LinearHomogeneousGaugeEvolEqs} is uniformly semi-normalized with respect to linear superposition of solutions. This is why all ``traditional'' Lie reductions for equations from this class are obtainable upon classifying $\pi_*G^\sim$-inequivalent subalgebras of~$\mathfrak g^{\rm ess}_\spanindex=\pi_*\mathfrak g^\sim$, which is equivalent to classifying inequivalent subalgebras of~$\mathfrak g^\sim$. 
(See also below for involving vector fields from~$\mathfrak g^{\rm lin}_\spanindex$ in the process of Lie reduction.)
Since equations from the class~\eqref{eq:LinearHomogeneousGaugeEvolEqs} have two independent variables, reductions of related equations with respect to one-dimensional subalgebras of~$\mathfrak g^{\rm ess}_\spanindex$ yield ordinary differential equations.

\begin{proposition}\label{pro:Inequiv1DSubalgOfEssSPanAlg}
 An optimal list of one-dimensional subalgebras of the algebra~$\mathfrak g^{\rm ess}_\spanindex$ is exhausted by the algebras
\begin{equation}\label{eq:ClassificationSubalgebrasLinearEvolEqs}
 \langle D(1)\rangle,\quad \langle P(1)+I(\phi)\rangle,\quad \langle I(\hat\phi)\rangle,
\end{equation}
where~$\phi=\phi(t)$ is an arbitrary smooth function of~$t$, and~$\hat\phi\in\{1,t\}$.
\end{proposition} 

\begin{proof}
 The most general element of the algebra~$\mathfrak g^{\rm ess}_\spanindex$ is of the form
$Q=D(\tau)+P(\chi)+I(\phi).$
 
 If~$\tau\ne0$ we can use the adjoint actions~$\mathcal D_*(T)$, $\mathcal P_*(X^0)$ and~$\mathcal I_*(U^1)$ given in~\eqref{eq:AdjointActionsEvolEqs} for suitable functions~$T$, $X^0$ and~$U^1$ to set~$\tau=1$, $\chi=0$ and~$\phi=0$.
 
 If $\tau=0$ but~$\chi\ne0$ we can use the adjoint action~$\mathcal D_*(T)$ to set $\chi=1$ and no further simplifications are possible, leaving~$\phi$ an arbitrary function of~$t$.
 
 If $\tau=\chi=0$ and~$\phi\ne0$, two possible cases arise. If~$\phi=\const$ then we can set~$\phi=1$ by scaling~$Q$. If~$\phi\ne\const$, then using the adjoint action~$\mathcal I_*(T)$ we can set~$\phi=t$. 
\end{proof}

The algebra~$\langle I(t)\rangle$ is not a Lie symmetry generator for an equation from the class~\eqref{eq:LinearHomogeneousGaugeEvolEqs}.
The algebra~$\langle I(1)\rangle$ cannot be used for Lie reductions as it does not allow one to make an ansatz for~$u$. This is why only the first two algebras given in~\eqref{eq:ClassificationSubalgebrasLinearEvolEqs} have to be considered.

\medskip

\noindent\textit{Algebra $\langle D(1)\rangle$}. 
This reduction is thus relevant for equations equivalent to 
equations given by Case~2 in Table~\ref{tab:CompleteGroupClassificationLinearEvolEqs} 
or its further Lie symmetry extensions, i.e., $A^l=A^l(x)\bmod G^\sim$, $l=0,\dots,\EqOrd-2$. 
Functionally independent invariants of the infinitesimal generator~$D(1)$ are~$\omega=x$ and $v=u$. 
The ansatz~$u=v(\omega)$ reduces a relevant equation~$\mathcal L_A$ of the form~\eqref{eq:LinearHomogeneousGaugeEvolEqs} 
to an (in general, variable-coefficient) $\EqOrd$th order linear ordinary differential equation in the rational form,
\[
\mathcal L_A^{\rm r}\colon\quad v_n+A^lv_l=0,
\]
where~$v_k=\mathrm{d}^k v/\mathrm{d}\omega^k$. 
The essential Lie invariance algebra of~$\mathcal L_A^{\rm r}$ consists of the vector fields 
of the form $D^{\rm r}(\zeta)+cv\p_v$, where 
$D^{\rm r}(\zeta)=\zeta\p_\omega+\frac12(\EqOrd-1)\zeta_\omega v\p_v$,  
$\zeta$ is a smooth function of~$\omega$ and $c$ is a constant;
see~\cite{boyk15a} and references therein.
Moreover, each vector field of the above form is a Lie symmetry generator 
of the equation~$\mathcal L_A^{\rm r}$ for certain values of~$A$, 
and it is induced by a Lie symmetry of the corresponding initial equation~$\mathcal L_A$ if and only if $\zeta_{xx}=0$.
Therefore, any essential Lie symmetry generator of~$\mathcal L_A^{\rm r}$ with $\zeta_{xx}\ne0$ 
can be considered as a so-called additional~\cite[Example~3.5]{olve86Ay} (hidden~\cite{abra08a}) symmetry of the equation~$\mathcal L_A$.

\medskip

\noindent\textit{Algebra $\langle P(1)+I(\phi)\rangle$}. 
The corresponding reduction is relevant for values of arbitrary elements 
$A^j=A^j(t)$, $j=2,\dots,\EqOrd-2$, $A^1=0$, $A^0=f(t)x$ ${}\bmod G^\sim$, 
where~$f=\phi_t$, as given for Case~4 in Table~\ref{tab:CompleteGroupClassificationLinearEvolEqs} 
and further cases of Lie symmetry extensions. 
Two functionally independent invariants of the vector field~$P(1)+I(\phi)$ are~$\omega=t$ and~$v=e^{-\phi x}u$. 
The solution ansatz~$u=e^{\phi x}v(\omega)$ then reduces the equation~$\mathcal L_A$ 
to a first-order linear ordinary differential equation,
\[
\mathcal L_A^{\rm r}\colon\quad v_\omega=(\phi^\EqOrd+A^j\phi^j)v,
\]
where the superscripts of~$\phi$ denote exponents, 
and the summation over the repeated index~$j$ from $2$ to $\EqOrd-2$ is implied. 
This equation is readily integrated and its general solution is
\[
v=c_0\exp\left(\int(\phi^\EqOrd+A^j\phi^j)\,{\rm d}\omega\right),
\]
where~$c_0=\const$.
Since the equation~$\mathcal L_A^{\rm r}$ is of order one, 
its maximal Lie invariance algebra is infinite-dimensional 
and consists of vector fields that are, in general, not fiber-preserving. 
Therefore, the equation~$\mathcal L_A$ admits infinitely many hidden symmetries related to this reduction. 

\begin{remark}
 The above two reductions seem quite elementary due to classifying reductions up to $G^\sim$-equivalence. 
 At the same time, reductions of specific equations with respect to their one-dimensional Lie symmetry algebras are generally more involved. 
\end{remark}

Each Lie symmetry vector field $Q\in\mathfrak g_A$ of the equation~$\mathcal L_A$ 
generates a one-parameter point-symmetry group of this equation. 
Point symmetry transformations of~$\mathcal L_A$ are of the form~\eqref{eq:PointTransformationBetweenEvolEqs}, 
where the function~$U^0$ runs through the solution set of~$\mathcal L_A$,
and $\tilde A^l(\tilde t, \tilde x)=A^l(\tilde t,\tilde x)$ in~\eqref{eq:PointTransformationBetweenEvolEqsAE1}. 
These transformations can be used for the construction of new solutions from known ones, including linear superposition of solutions. 
For~$\mathcal L_A$ as linear equation, Lie symmetry vector fields can also be used directly 
for the same purpose in the following way~\cite[Section~7.14]{ovsi82Ay}: 
Given a vector field $Q=\tau\p_t+\xi\p_x+(\phi u+\eta^0)\p_u\in\mathfrak g_A$ and a known solution~$u=h(t,x)$ of~$\mathcal L_A$, 
the function $u=Q[h]:=\phi h+\eta^0-\tau h_t-\xi h_x$ is also a solution of~$\mathcal L_A$.
Proposition~\ref{pro:Inequiv1DSubalgOfEssSPanAlg} implies that, 
up to $G^\sim$-equivalence and for relevant values of~$A$, 
such nontrivial actions on solutions are exhausted by $u\to -u_t$ and $u\to\phi(t)u-u_x$.

A new procedure for generating exact solutions of linear systems using their Lie symmetries was suggested in~\cite{fush96a}. 
It can be considered as the inversion of the above actions on solutions by Lie symmetry vector fields. 
It can also be interpreted as the Lie reduction that involves vector fields associated with linear superposition of solutions.  
Later, a particular case of this procedure was considered in~\cite{broa99a}. 
Given a nonzero vector field $Q=\tau\p_t+\xi\p_x+\phi u\p_u\in\mathfrak g^{\rm ess}_A$ 
and a known solution~$u=h(t,x)$ of~$\mathcal L_A$ (therefore, $h\p_u\in\mathfrak g^{\rm lin}_A$), 
consider the vector field $Q_h=Q+h\p_u\in\mathfrak g_A$. 
Looking for $Q_h$-invariant solutions of~$\mathcal L_A$, 
we first solve the equation $\phi u-\tau u_t-\xi u_x=-h$ to construct an ansatz for~$u$, 
reduce the equation~$\mathcal L_A$ by this ansatz
and then integrate the obtained reduced equation~$\mathcal L_A^{\rm r}$, 
which is an inhomogeneous linear ordinary differential equation.  
In view of Proposition~\ref{pro:Inequiv1DSubalgOfEssSPanAlg}, 
there are two $G^\sim$-inequivalent realizations of this procedure 
for relevant equations from the class~\eqref{eq:LinearHomogeneousGaugeEvolEqs},

\bigskip\par\noindent
$A^l=A^l(x),\ A^{\EqOrd-1}=0,\ A^\EqOrd=1;\quad Q=D(1)\colon$
\begin{gather*}
u=v(\omega)+\int_{t_0}^th(t',x)\,{\rm d}t',\ \omega=x,\quad
\mathcal L_A^{\rm r}\colon\quad 
v_n+\sum_{l=0}^{\EqOrd-2}A^lv_l=h\big|_{t=t_0};
\end{gather*}

\noindent
$A^0=\phi_t(t)x,\ A^1=0,\ A^j=A^j(t),\ j=2,\dots,\EqOrd-2,\ A^{\EqOrd-1}=0,\ A^\EqOrd=1;\quad Q=P(1)+I(\phi)\colon$
\begin{gather*}
u=e^{\phi(t)x}v(\omega)+e^{\phi(t)x}\int_{x_0}^xe^{-\phi(t)x'}h(t,x')\,{\rm d}x',\ \omega=t, 
\\ 
\mathcal L_A^{\rm r}\colon\quad 
v_\omega=\left(\,\sum_{k=2}^\EqOrd A^k\phi^k\right)v+e^{-\phi x_0}\sum_{k=2}^\EqOrd A^k\sum_{i=0}^{k-1}h_i\big|_{x=x_0}\phi^{k-1-i},
\end{gather*}
where the superscripts of~$\phi$ again denote exponents. 
Integrating the reduced equation~$\mathcal L_A^{\rm r}$ for the second realization, 
we derive an explicit nonlocal formula for generating a new solution from 
a known one, 
\begin{gather*}
u=e^{\phi(t)x}\int_{x_0}^xe^{-\phi(t)x'}h(t,x')\,{\rm d}x'
\\
\phantom{u=}{}+
\left(
\int_{t_0}^t\sum_{k=2}^\EqOrd A^k(t')\sum_{i=0}^{k-1}\big(\phi(t')\big)^{k-i-1}h_i(t',x_0)
e^{-\phi(t')x_0-w(t')}\,{\rm d}t'+v_0
\right)e^{\phi(t)x+w(t)}, 
\end{gather*}
where $\displaystyle w(t):=\int_{t_0}^t\sum_{k=2}^\EqOrd A^k(t')\big(\phi(t')\big)^k\,{\rm d}t'$.

The above procedure can be iterated using various~$Q$'s. 
In particular, starting with $h=0$ and using the same $Q=D(1)$ for each iteration, 
we construct solutions that are polynomials in~$t$ with coefficients depending on~$x$, 
\begin{equation}\label{eq:ReductionOfLinEvolEqsWRT_D(1)^N}
u=\sum_{s=0}^N v^s(x)t^s, \quad\mbox{where}\quad
v^s_\EqOrd+\sum_{l=0}^{\EqOrd-2}A^lv^s_l=(s+1)v^{s+1},\ s=0,\dots,N,\quad v^{N+1}:=0.
\end{equation}

Another use of Lie symmetries of a linear system of differential equations 
for finding its exact solutions   
is based on generating higher-order (generalized) infinitesimal symmetries of this system 
by its Lie symmetries \cite[Proposition 5.22]{olve86Ay}.
For any generalized symmetry of a system~$\mathcal L$ of differential equations, 
the associated invariant surface condition (i.e., the condition of vanishing of its characteristic) 
is a differential constraint formally consistent with the system~$\mathcal L$.
Given an equation~$\mathcal L_A$ from the class~\eqref{eq:LinearHomogeneousGaugeEvolEqs} 
with $\mathfrak g^{\rm ess}_A\ne\{0\}$, 
each nonzero vector field $Q=\tau\p_t+\xi\p_x+\phi u\p_u\in\mathfrak g^{\rm ess}_A$ 
corresponds to the recursion operator $\mathrm Q=-\tau\mathrm D_t-\xi\mathrm D_x+\phi$ 
of the equation~$\mathcal L_A$, i.e., 
the operator~$\mathrm Q$ maps the set of generalized-symmetry characteristics of~$\mathcal L_A$ 
to itself. 
Subsequently acting by recursion operators associated with various basis elements of~$\mathfrak g^{\rm ess}_A$ on the characteristic~$u$ of the trivial infinitesimal symmetry~$I(1)$ 
and linearly combining results of such actions, 
one obtains the set of characteristics of ``linear'' generalized symmetries generated by Lie symmetries. 
If a generalized symmetry is generated by iterative use of a single Lie symmetry in the above way, 
the general solution of the corresponding invariant surface condition, which is in fact equivalent to a linear ordinary differential equation with constant coefficients, gives 
an ansatz with several new unknown functions of a single new independent variable. 
This ansatz reduces the initial equation to a linear system of ordinary differential equations 
for the new unknown functions. 
The construction of solutions of the form~\eqref{eq:ReductionOfLinEvolEqsWRT_D(1)^N} can be interpreted 
as example of such generalized reductions,  
where $Q=D(1)$, $\mathrm Q=-\mathrm D_t$ and the invariant surface condition is $\mathrm Q^{N+1}u=0$. 
In general, essential (up to $G^\sim$-equivalence and linear superposition of solutions) 
generalized Lie-symmetry-related reductions of equations from the class~\eqref{eq:LinearHomogeneousGaugeEvolEqs} 
are exhausted by the reductions presented below, 
which arise from the factorization of general polynomials of~$\mathrm Q$ 
over the corresponding underlying field (complex or real).
For each case we list the corresponding Lie symmetry vector field, recursion operator, 
invariant surface condition and reduced system of ordinary differential equations. 
Here $N\in\mathbb N_0$, 
the index~$s$ runs from~0 to~$N$, 
the index~$l$ runs from~0 to $\EqOrd-2$,
and we assume the summation over the repeated indices.
The constant~$\lambda$ is from the underlying field ($\mathbb R$ or~$\mathbb C$),  
$\mu$ and $\nu$ are real constants, and $\nu>0$.

\bigskip\par\noindent
$A^l=A^l(x),\quad A^{\EqOrd-1}=0,\quad A^\EqOrd=1$:
\begin{gather*}
Q=D(1),\quad \mathrm Q=-\mathrm D_t,\quad (\mathrm Q+\lambda)^{N+1}u=0,
\\[1ex] \qquad
u=v^s(x)t^se^{\lambda t},\\[1ex]\qquad
v^s_\EqOrd+A^lv^s_l=(s+1)v^{s+1}+\lambda v^s,\quad v^{N+1}:=0;
\\[1.5ex]
(\mbox{over }\mathbb R)\quad 
Q=D(1),\quad \mathrm Q=-\mathrm D_t,\quad ((\mathrm Q+\mu)^2+\nu^2)^{N+1}u=0,
\\[1ex] \qquad
u=\big(v^s(x)\cos(\nu t)+w^s(x)\sin(\nu t)\big)t^se^{\mu t}, 
\\[1ex] \qquad
v^s_\EqOrd+A^lv^s_l=(s+1)v^{s+1}+\mu v^s+\nu w^s,\quad v^{N+1}:=0,
\\ \qquad
w^s_\EqOrd+A^lw^s_l=(s+1)w^{s+1}-\nu v^s+\mu w^s,\quad w^{N+1}:=0.
\end{gather*}

\medskip\par\noindent
$A^0=\phi_t(t)x,\quad A^1=0,\quad A^j=A^j(t),\ j=2,\dots,\EqOrd-2,\quad A^{\EqOrd-1}=0,\quad A^\EqOrd=1$:
\begin{gather*}
Q=P(1)+I(\phi),\quad \mathrm Q=-\mathrm D_x+\phi,\quad (\mathrm Q+\lambda)^{N+1}u=0,
\\ \qquad
u=v^s(t)x^se^{(\phi+\lambda)x},\quad
v^s_t=\sum_{k=2}^\EqOrd A^k\sum_{p=s}^{\min(N,k+s)}\binom k{p-s}\frac{p!}{s!}(\phi+\lambda)^{k+s-p}v^p;
\\[1.5ex]
(\mbox{over }\mathbb R)\quad 
Q=P(1)+I(\phi),\quad \mathrm Q=-\mathrm D_x+\phi,\quad ((\mathrm Q+\mu)^2+\nu^2)^{N+1}u=0,
\\[1ex] \qquad
u=\big(v^s(t)\cos(\nu x)+w^s(t)\sin(\nu x)\big)x^se^{(\phi+\mu)x}, 
\\[1ex] \qquad
v^s_t=\sum_{k=2}^\EqOrd A^k\sum_{p=s}^{\min(N,k+s)}\binom k{p-s}\frac{p!}{s!}(\Phi^{skp}v^p+\Psi^{skp}w^p),
\\ \qquad
w^s_t=\sum_{k=2}^\EqOrd A^k\sum_{p=s}^{\min(N,k+s)}\binom k{p-s}\frac{p!}{s!}(-\Psi^{skp}v^p+\Phi^{skp}w^p),
\\[1ex] \qquad\mbox{where}\quad 
\Phi^{skp}=\sum_{q=0}^{\left[\frac{k+s-p}2\right]}(-1)^q\binom{k+s-p}{2q}\nu^{2q}(\phi+\mu)^{k+s-p-2q},
\\ \qquad\phantom{\mbox{where}}\quad 
\Psi^{skp}=\sum_{q=0}^{\left[\frac{k+s-p-1}2\right]}(-1)^q\binom{k+s-p}{2q+1}\nu^{2q+1}(\phi+\mu)^{k+s-p-2q-1}.
\end{gather*}
For the above reductions, the parameters~$\lambda$ and~$(\mu,\nu)$ can be set to $0$ and $(0,1)$ 
by $\mathcal I(e^{-\lambda t})$ and $\mathcal D(\nu t)\mathcal I(e^{-\mu t})$ if $\mathrm Q=-\mathrm D_t$
or by re-denoting $\phi+\lambda\to\phi$ and $\phi+\mu\to\phi$ 
with sequentially applying $\mathcal D(\nu^\EqOrd t)$ and/or $\mathcal X$ if $\mathrm Q=-\mathrm D_x+\phi$, 
respectively.
At the same time, the parameters become essential 
when an arbitrary polynomial~$\mathrm P$ of~$\mathrm Q$ is considered.
The reduction with respect to~$\mathrm P$ gives solutions 
that are linear superpositions of solutions constructed for the corresponding maximal powers 
of irreducible multipliers of~$\mathrm P$.

All the above methods for finding exact solutions can be combined with each other.

\section{Conclusion}\label{sec:ConclusionLinearEvolutionEqs}

In this paper we have exhaustively solved the group classification problem 
for the class~\eqref{eq:GeneralLinearEvolutionEquation} of (1+1)-dimensional linear evolution equations 
of arbitrary fixed order $\EqOrd>2$. 
This classification completes (and enhances) the previous group classifications 
of (1+1)-dimensional linear evolution equations of order three and four 
that were tackled in~\cite{gung04a,huan12a}. 
We have computed the equivalence groupoids and equivalence groups 
of the class~\eqref{eq:GeneralLinearEvolutionEquation}, 
of its nested subclasses defined by the gauges $A^\EqOrd=1$ and $(A^\EqOrd,A^{\EqOrd-1})=(1,0)$ 
and of the corresponding counterparts consisting of homogeneous equations. 
Although the above classes of, in general, inhomogeneous equations are normalized, 
the subclass that is the most convenient for group classification 
and that is minimal among subclasses 
whose group classifications are equivalent to the group classification of the entire class~\eqref{eq:GeneralLinearEvolutionEquation} 
is the subclass of linear homogeneous evolution equations of the reduced form~\eqref{eq:LinearHomogeneousGaugeEvolEqs}.
The associated gauge of arbitrary elements is $(A^\EqOrd,A^{\EqOrd-1},B)=(1,0,0)$. 
The fact that the subclass~\eqref{eq:LinearHomogeneousGaugeEvolEqs} 
is uniformly semi-normalized with respect to linear superposition of solutions 
allows us to use a special version of the algebraic method~\cite{kuru16a} 
in order to solve the group classification problem 
for the class~\eqref{eq:GeneralLinearEvolutionEquation} in the optimal way.
Due to the same fact, Lie invariant solutions of equations 
from the subclass~\eqref{eq:LinearHomogeneousGaugeEvolEqs} 
as well as other solutions that are constructed by symmetry-based methods specific for linear equations 
can be classified up to $G^\sim$-equivalence, 
which is discussed in Section~\ref{sec:ExactSolutionsLinEvolEqs}.

It is worth mentioning the difference in the above version of the algebraic method to the classification technique employed in~\cite{gung04a,huan12a}, 
which was proposed in~\cite{zhda99Ay} and applied therein to the group classification of a class of second-order nonlinear evolution equations. 
This technique is implicitly based on the normalization property of a class~$\mathcal L|_{\mathcal S}$ of (systems of) differential equations to be classified, 
and its main steps are the following:
\begin{enumerate}\itemsep=0ex
\item
Compute the equivalence group~$G^\sim$ of~$\mathcal L|_{\mathcal S}$ 
and the span~$\mathfrak g_\spanindex$ of Lie symmetry algebras of equations from the class~$\mathcal L|_{\mathcal S}$.
\item
Construct an optimal list of certain $G^\sim$-equivalent low-dimensional subalgebras 
of~$\mathfrak g_\spanindex$, 
find the subclasses of equations that possesses these subalgebras as their Lie symmetry algebras, and show that up to~$G^\sim$-equivalence 
each higher-dimensional extension is contained in a subclass 
associated with a listed subalgebra of dimension maximal among considered ones.
\item
Solve the particular group classification problem for each of the selected subclasses.
\item
Compose a~classification list for the class~$\mathcal L|_{\mathcal S}$ 
as the union of the lists constructed in the course of solving the above particular problems. 
\end{enumerate}
Since the number of arguments in arbitrary elements parameterizing a selected subclass 
is usually less than that for the class~$\mathcal L|_{\mathcal S}$, 
the group classification problems for subclasses have more chances 
to be solved by the direct integration of the corresponding systems of determining equations 
for the components of Lie symmetry vector fields. 
A~difficulty in implementing the described procedure is created by the appearance of $G^\sim$-equivalent cases in the course of classifying different subclasses. 
Thus, a part of cases that are $G^\sim$-equivalent to other listed cases 
should be excluded from the final joint list. 
At the same time, the use of equivalences in the course of directly solving determining equations 
is more delicate and complicated than within the framework of the algebraic method of group classification. 
The complexity of applying the technique led to weaknesses of derived classifications 
such as the weaknesses of~\cite{gung04a,huan12a} 
discussed in Remark~\ref{rem:ComparisonWithThreeFourOrderCase}.
The first three steps of the classification procedure by~\cite{zhda99Ay} are 
rather simple for (1+1)-dimensional linear evolution equations
but this is not the case for the fourth step. 
In particular, merely one-dimensional subalgebras 
of the span of the essential Lie invariance algebras 
of linear homogeneous evolution equations from the reduced subclass~\eqref{eq:LinearHomogeneousGaugeEvolEqs} 
were classified in the second step in~\cite{gung04a,huan12a}. 
Note that the sufficiency of classifying the essential Lie invariance algebras 
for classes of linear homogeneous equations was not justified in~\cite{gung04a,huan12a}; 
cf.~\cite{kuru16a}.  

In turn, in the present paper we first establish the uniform semi-normalization of the subclass~\eqref{eq:LinearHomogeneousGaugeEvolEqs} with respect to linear superposition of solutions, 
followed by classifying all subalgebras of the projection~$\pi_*\mathfrak g^\sim$
of the equivalence algebra of the subclass~\eqref{eq:LinearHomogeneousGaugeEvolEqs}
that are appropriate for the solution of the corresponding group classification problem.
The classification of such subalgebras is possible due to 
the preceding derivation of constraints for them 
in the series of Lemmas~\ref{lem:DimOfEssLieInvAlgebra}--\ref{lem:OnOperatorsInvolvingTau}, 
including the least upper bound for their dimension, which equals four. 
Among the constraints for appropriate subalgebras
there are also inequalities for the dimensions of several specific subspaces 
of each appropriate subalgebra. 
These dimensions are $G^\sim$-invariant small nonnegative integers 
and are hence convenient for the marking of classification cases. 
The collection of constraints presented in
Lemmas~\ref{lem:DimOfEssLieInvAlgebra}--\ref{lem:OnOperatorsInvolvingTau} completely defines the set of appropriate subalgebras 
since any subalgebra of~$\pi_*\mathfrak g^\sim$ satisfying these constraints 
is the maximal Lie invariance algebra for an equation of the form~\eqref{eq:LinearHomogeneousGaugeEvolEqs}. 
After constructing an optimal list of appropriate subalgebras of~$\pi_*\mathfrak g^\sim$,
for each subalgebra~$\mathfrak s$ from this list
we substitute the components of basis elements of~$\mathfrak s$ 
into the classifying equations~\eqref{eq:DeterminingEqsLinearEvolEqs1}--\eqref{eq:DeterminingEqsLinearEvolEqs3} 
and integrate the obtained system on the tuple~$A$. 
The obtained general solution should be gauged by elements of 
the stabilizer subgroup of~$G^\sim$ with respect to~$\mathfrak s$.
This strategy avoids arising equivalent cases in the course of group classification.

As referenced through the present paper, 
there are several similarities of this paper with the recent work~\cite{kuru16a} 
on group classification of the class of (1+1)-dimensional linear Schr\"odinger equations with complex potential. 
Both the class~\eqref{eq:LinearHomogeneousGaugeEvolEqs} and the class studied in~\cite{kuru16a} 
are uniformly semi-normalized with respect to linear superposition of solutions. 
The group classification problems for both the classes are hence solved 
by studying certain low-dimensional subalgebras of the projections of the corresponding equivalence algebras. 
At the same time, the preceding description of properties of appropriate subalgebras in~\cite{kuru16a} 
is not complete in contrast to the present paper, 
and thus in~\cite{kuru16a} there is an additional selection of appropriate subalgebras 
in the course of their classification. 

The consideration of second-order evolution equations 
within the framework of group analysis of differential equations 
was recently extended with the study of their reductions operators~\cite{popo08a}, 
local and potential conservation laws and potential symmetries~\cite{popo08Ay}.
The same study may be carried out for an arbitrary order $\EqOrd>2$. 
Some preliminary results were recently obtained in this direction. 
In particular, local conservation laws and simplest potential conservation laws 
of equations from the class~\eqref{eq:GeneralLinearEvolutionEquation} 
were described in~\cite{popo10d} and in~\cite{boyk11a}, respectively. 

\medskip

\noindent{\bf Acknowledgments.}
The authors thank the anonymous reviewer for several helpful remarks.
This research was undertaken, in part, thanks to funding from the Canada Research Chairs program 
and the NSERC Discovery Grant program. 
The research of ROP was supported by the Austrian Science Fund (FWF), project P25064.

\footnotesize\setlength{\itemsep}{0ex}

\end{document}